\documentclass[journal, onecolumn, 11pt]{IEEEtran}

\usepackage{setspace}
\usepackage{amsmath}
\usepackage{amssymb}
\usepackage{mathrsfs}
\usepackage{amsthm}
\usepackage{multirow}
\usepackage{color}
\usepackage{array}
\usepackage{url}
\usepackage{comment}
\usepackage{enumerate}
\usepackage{eucal}

\usepackage{algorithm}
\usepackage[noend]{algorithmic}
\usepackage{cite}
\usepackage[table]{xcolor}
\usepackage{blkarray}

\usepackage{stmaryrd}

\usepackage[top=0.72in, bottom=1in, left=1in, right=1in]{geometry}

\usepackage{tikz}
\usetikzlibrary{arrows, patterns, shapes.arrows, decorations.pathmorphing}

\IEEEoverridecommandlockouts

\theoremstyle{plain}
\newtheorem{theorem}{Theorem}
\newtheorem{corollary}[theorem]{Corollary}
\newtheorem{lemma}[theorem]{Lemma}
\newtheorem{proposition}[theorem]{Proposition}

\theoremstyle{definition}
\newtheorem{definition}{Definition}
\newtheorem{example}[definition]{Example}

\newtheorem{remark}[definition]{Remark}

\newcommand{\bbS}{\mathbb{S}}
\newcommand{\bbT}{\mathbb{T}}
\newcommand{\bbV}{\mathbb{V}}
\newcommand{\bbE}{\mathbb{E}}
\newcommand{\bbL}{\mathbb{L}}
\newcommand{\bbP}{\mathbb{P}}

\DeclareMathAlphabet{\mathbfsl}{OT1}{ppl}{b}{it} 

\newcommand{\vp}{\mathbfsl{p}}

\newcommand{\vx}{\mathbfsl{x}}
\newcommand{\vy}{\mathbfsl{y}}
\newcommand{\vz}{\mathbfsl{z}}

\newcommand{\vT}{\mathbfsl{T}}
\newcommand{\vU}{\mathbfsl{U}}
\newcommand{\TBSCu}{\mathbfsl{T}_{\rm BSC}^{(u)}}
\newcommand{\TBSCl}{\mathbfsl{T}_{\rm BSC}^{(\ell)}}
\newcommand{\TBEC}{\mathbfsl{T}_{\rm BEC}}

\newcommand{\cC}{\mathcal{C}}

\newcommand{\cG}{\mathcal{G}}
\newcommand{\cH}{\mathcal{H}}
\newcommand{\cM}{\mathcal{M}}
\newcommand{\cO}{\mathcal{O}}

\newcommand{\perm}{{\rm per}}
\newcommand{\val}{{\rm val}}
\newcommand{\permsec}{{\rm per}^{(2)}}
\newcommand{\perror}[1]{P_{\rm error}\left( #1 \right)}
\newcommand{\prob}[1]{{\rm Prob}\left( #1 \right)}

\newcommand{\id}{{\rm id}}
\DeclareMathOperator*{\argmax}{arg\,max}

\newcommand{\etal}{{\em et al.}}

\newcommand{\todo}[1]{{\color{red} TODO: #1}}

\title{Efficient Algorithms for\\ the Bee-Identification Problem}

\author{
 \IEEEauthorblockN{
	Han Mao Kiah\IEEEauthorrefmark{1},
	Alexander Vardy\IEEEauthorrefmark{2}, and
	Hanwen Yao\IEEEauthorrefmark{2}\\
	}
	\IEEEauthorblockA{\IEEEauthorrefmark{1}\small School of Physical and Mathematical Sciences, Nanyang Technological University, Singapore}\\ 
   \IEEEauthorblockA{\IEEEauthorrefmark{2}\small Department of Electrical \& Computer Engineering, University of California San Diego, LA Jolla, CA, USA}\\
 \IEEEauthorblockA{\small Emails: 
 	hmkiah@ntu.edu.sg, avardy@ucsd.edu, hay125@eng.ucsd.edu}
\thanks{Parts of this work were presented in the IEEE International Symposium on Information Theory (ISIT2021) \cite{Kiah2021}.} 
}

\begin{document}
\date{}

\maketitle

\hspace*{-10pt}\begin{abstract}
The bee-identification problem, formally defined by Tandon, Tan and Varshney (2019), requires the receiver to identify ``bees'' using a set of unordered noisy measurements.
In this previous work, Tandon, Tan and Varshney studied error exponents and showed that decoding the measurements {\em jointly} results in a significantly smaller error exponent.

In this work, we study algorithms related to this joint decoder.
First, we demonstrate how to perform joint decoding efficiently.
By reducing to the problem of finding perfect matching and minimum-cost matchings, 
we obtain joint decoders that run in time quadratic and cubic in the number of ``bees'' for the binary erasure (BEC) and binary symmetric channels (BSC), respectively.
Next, by studying the matching algorithms in the context of channel coding, we further reduce the running times by using classical tools like peeling decoders and list-decoders.
In particular, we show that our identifier algorithms when used with Reed-Muller codes terminates in almost linear and quadratic time for BEC and BSC, respectively.

Finally, for explicit codebooks, we study when these joint decoders fail to identify the ``bees'' correctly. 
Specifically, we provide practical methods of estimating the probability of erroneous identification for given codebooks.
\end{abstract}


\section{Introduction}\label{sec:intro}

Imagine $M$ bees, each tagged with a unique barcode, flying in a beehive. 
We take a picture of the bees and obtain an unordered set of noisy barcodes.
The {\em bee-identification problem} -- proposed and formally defined by Tandon \etal -- requires one to uniquely identify each bee from the noisy measurements \cite{TTV2019}.
Besides problems involving multiple target tracking \cite{Poore2006,Gernat2018}, the bee-identification problem is also relevant to other applications (see \cite{TTV2019, Pananjady2018} for other examples). 
One recent possible application is that of pooled testing for viral RNA like COVID-19. 
In a recent experiment \cite{Schmid2020}, Schmid-Burgk \etal{} developed a procedure where multiple DNA samples are pooled, sequenced and analyzed {\em en masse} for the COVID-19 infection. In their procedure, barcodes with high Levenshtein distance were inserted in the DNA samples and by decoding the barcodes individually, they were able to identify the viral DNA samples. Later, the procedures were validated by other groups who performed similar experiments \cite{Li2020, James2020, Peto2020, Booeshaghi2020}.

Indeed, to recover the original barcodes, a naive approach is to look at each barcode separately and decode them {\em independently}.
However, certain bees/DNA samples may be assigned to the same barcode and in this case, we fail to identify all the bees/DNA samples.
In contrast, one can look at all the barcodes {\em jointly} and determine the best way to assign the barcodes so that likelihood of correct identification is maximized.
The latter is termed as joint decoding and 
in \cite{TTV2019}, Tandon \etal{} showed that joint decoding results in significantly smaller probability of wrong or failed identification. Specifically, they quantified the gap between the error exponents of independent and joint decoding.
Interestingly, in a follow up work, Tandon \etal{} showed that the error exponents are the same for both independent and joint decoding when bees are absent with certain probability \cite{TTV2020}.

In \cite{TTV2019}, Tandon \etal{} wrote that the lower error exponent of joint decoding comes at a ``cost of increased computational complexity''. 
They then posited that joint decoding entails a computationally prohibitive exhaustive search amongst the $M!$ possible permutations and explored ideas that combine both independent and joint decoding.
 
Fortunately, an exhaustive search is not necessary and
in this work, we demonstrate that {\em efficient joint decoding is achievable}. Specifically, for the binary erasure and binary symmetric channels, we reduce the bee-identification problem to the problem of finding a perfect matching and minimum-cost matching, respectively.
Hence, applying the well-known Hopcraft-Karp algorithm \cite{HopcroftKarp1973} and Hungarian method \cite{Kuhn1955}, respectively, we can identify the bees in time polynomial in $M$.

We then study the (minimum-cost) matching problem in the context of channel coding and show that the complexity of bee-identification problem can be further reduced. 
In particular, for the binary erasure channel, we showed that 
when we deploy the celebrated Reed-Muller codes, 
the bee-identification problem can be resolved in {\em almost $O(M)$ time on average}. This is essentially optimal as $\Omega(M)$ time is required to read all $M$ barcodes. 
Therefore, not only is the ``cost of increased computational complexity'' for joint decoding acceptable, but, in some cases, the additional complexity cost is negligible.

Finally, we also investigate the probability of erroneous identification for these joint decoders. 
Specifically, in this work, by relating these probability computations to the problem of permanent computation, we develop practical methods of analyzing and estimating these error probabilities for any {\em code of interest}. 
In contrast, in \cite{TTV2020}, Tandon \etal{} fixed a certain code rate $R>0$ and determined a corresponding probability estimate $p(R)$. Then using random coding techniques, they showed the existence of a code whose rate is approximately $R$ with the property: 
under joint decoding, the probability of erroneous identification is at most $p(R)$.

In the next section, we formally describe the bee-identification problem given in \cite{TTV2020} and then state our technical contributions.
For the ease of exposition, we study the bee-identification problem for the binary erasure channel (BEC) and binary symmetric channel (BSC). 
The methods in this paper can be extended for larger alphabets and 
to perform joint maximum-likelihood decoding for other channels.

Finally, to conclude this introduction, we mention certain work that followed the conference version of this work. In \cite{Chrisnata2022}, motivated by applications that involve DNA strands, Chrisnata \etal{} studied a version of the bee-identification problem where multiple outputs (from a single input) are available. 
In the same paper, Chrisnata \etal{} also studied the bee-identification problem in the context of deletion channels.
 

\newpage

\section{Problem Definition}

For an integer $M$,  we let $[M]$ denote the set of integers $\{1,2,\ldots, M\}$.
The set of all permutations over $[M]$ is denoted by $\bbS_M$ and 
we write a permutation $\sigma\in \bbS_M$ as $\sigma(1)\sigma(2)\cdots \sigma(M)$.

Consider a binary code $\cC\subseteq \{0,1\}^n$ of length $n$ with  $M$ {\em codewords} $\vx_1,\vx_2,\ldots, \vx_M$. 
Consider, in addition, a binary channel where each output $\vy$ given an input $\vx$ is received with probability $P(\vy|\vx)$.
We send {\em all} $M$ codewords over the channel and obtain an unordered set of $M$ outputs $\{\vy_1, \vy_2, \ldots , \vy_M\}$.
Note that $\vy_i$ is not necessarily the channel output of $\vx_i$ and in fact, the task of the {\em bee-identification problem} is to find a length-$M$ permutation $\sigma$ so that $\vy_{\sigma(i)}$ is indeed the channel output of the input $\vx_i$ for all $i\in[M]$. 
Assuming the channels are independent, the {\em joint decoder} finds a length-$M$ permutation $\sigma^*$ that maximizes the probability 
$\prod_{i\in [M]} P(\vy_{\sigma(i)}| \vx_i)$. In other words, the joint decoder returns a permutation $\sigma^*$ such that 
\[ \sigma^*\in \argmax_{\sigma\in\bbS_M} \prod_{i\in [M]} P(\vy_{\sigma^*(i)}| \vx_i)\,.\]

In this paper, we first study efficient ways of performing joint decoding, that is, computing the permutation $\sigma^*$.
Since the input to our problem is a set of $M$ $n$-bit codewords, 
a trivial lower bound on complexity is $\Omega(Mn)$ and 
our running time analysis in most parts will be with respect to the parameter $M$.
Also, as the code size $M$ represents ``the number of bees'', 
we assume a reasonable growth rate of $M$ with respect to $n$, that is, polynomial in $n$. Hence, in most parts of the paper, we suppress factors involving $n$ in the big-O notation.
We note that this is somewhat different from the setting in \cite{TTV2019} where 
$M=2^{Rn}$ for some positive rate $R$. 

Next, for a fixed code $\cC$, we also investigate when the joint decoder fails to return the correct permutation $\sigma$. 
Namely, if $\sigma^*$ is the permutation returned by the joint decoder,  we provide estimates on the quantity, $\perror{\cC}$, the probability of the event where $\sigma\ne \sigma^*$.
Similar estimates were given in \cite{TTV2020}.
Specifically, for a fixed value $0<R<1$, Tandon \etal{} found an exponent $e(R)>0$ such that the following holds: 
there exists a family of codes $\{\cC_n\}_{n\ge 1}$ with blocklengths $n$ and rates approaching $R$ so that $\perror{\cC_n}\le 2^{-e(R)n}$.
As their derivations rely on random coding techniques, it is unclear whether their estimates apply to explicit codes. 
In contrast, we fix a specific code $\cC$ in this work and provide practical methods of estimating $\perror{\cC}$.

\subsection{Our Contributions}

We summarize our contributions here. 




\begin{itemize}
	\item For the BEC, we provide a joint decoder -- Joint Erasure Decoding Identifier (JEDI) -- that runs in $O(M^2)$ time. For the family of $r$-th order Reed-Muller codes and any small $\epsilon$, we show that on average, JEDI terminates in $O(M^{1+\epsilon})$ time when $r\ge 2$ and in $O(M^{2+\epsilon})$ time when $r=1$.
	
	\item For the BSC, we provide a joint decoder -- Joint Minimum-Distance Decoding Identifier (JMDI)-- that runs in $O(M^3)$ time. 
	To improve the running time, we approximate the exact solution using ideas from list-decoding and propose the Joint List Decoding Identifier (JLDI). 
	For the family of $r$-th order Reed-Muller codes, we show that for sufficiently  small crossover probability $p$, JLDI terminates in $O(M^{2+\epsilon})$ time for any small $\epsilon$ and is almost as good as JMDI (see Theorem~\ref{thm:rm-jldi} for the formal statement).
	
	\item Finally, for a fixed code $\cC$, we provide probability estimates on when our joint decoders are erroneous. Specifically, using trellis-based techniques, we provide methods to compute upper and lower bounds for error probability in $O(M2^{M})$ and $O(M^{3/2}2^{2M})$ time, respectively. We also derive a closed formula that computes a weaker upper bound efficiently. 
\end{itemize}

\section{Joint Erasure Decoding Identifier}
\label{sec:jedi}

In this section, we consider the binary erasure channel (BEC).
Even though the case for BECs was not studied in \cite{TTV2019}, 
we investigate the joint decoder for the erasure channel as 
it illustrates certain key graph theoretic concepts for the Joint Minimum-Distance Decoding Identifier described in Section~\ref{sec:jmdi}.

Given an integer $M$, a {\em balanced bipartite graph $\cG$ of order $M$} is an undirected graph with $2M$ nodes: $M$ {\em left} and $M$ {\em right} nodes, where every edge connects a left node to a right node.
A {\em matching} $\cM$ of $\cG$ is a subset of edges where no two edges are incident on the same node. 
Clearly, any matching of a balanced bipartite graph $\cG$ of order $M$ has at most $M$ edges. 
If a matching $\cM$ contains exactly $M$ edges, we say that $\cM$ is {\em perfect}. 

Let us label the left and right nodes of a balanced bipartite graph $\cG$ of order $M$ with the $M$ inputs $\vx_1,\vx_2,\ldots, \vx_M$ and the $M$ outputs $\vy_1, \vy_2, \ldots , \vy_M$, respectively.
Suppose that we have a perfect matching $\cM$ of $\cG$. 
Then we can write the edges of $\cM$ as 
$(x_1, y_{\sigma(1)}), (x_2, y_{\sigma(2)}),\ldots, (x_M, y_{\sigma(M)})$ and
it follows from the definition of a matching that $\sigma$ is a permutation of length $M$. In other words, we can represent a perfect matching with a length-$M$ permutation. Conversely, given a length-$M$ permutation $\sigma$,
we obtain a perfect matching of $\cG$ if $(\vx_i,\vy_{\sigma(i)})$ is an edge of $\cG$ for all $i\in[M]$. Therefore, for the rest of this paper, we use permutations and matchings interchangeably.

We are now ready to describe the main contribution of this section: an efficient implementation of a joint decoder for erasures.

\vspace{3mm}

\noindent{\bf Joint Erasure Decoding Identifier (JEDI)}.

\noindent{\sc Input}: 
A codebook $\cC = \{\vx_1,\vx_2,\ldots, \vx_M\}\subseteq\{0,1\}^n$ of size $M$ and 
a set of $M$ channel outputs $\{\vy_1, \vy_2, \ldots , \vy_M\} \subseteq\{0,1,?\}^n$.\\[1mm]
\noindent{\sc Output}: A permutation $\sigma$ such that $\vy_i$ matches $\vx_{\sigma(i)}$ for all $i\in [M]$ if 
there is a unique $\sigma$. Otherwise, the decoder declares {\sc Failure}. 
\begin{enumerate}[(1)]
	\item We draw a balanced bipartite graph $\cG$ of order $M$.
	Here, the $M$ codewords are the left nodes while the $M$ channel outputs are the right nodes.
	For $i,j \in [M]$, we draw an edge between $\vx_i$ and $\vy_j$ if and only if $\vy_j$ {matches} $\vx_i$. Here, we say that $\vy$ {\em matches} $\vx$ if both $\vy$ coincides with $\vx$ on positions that are not erased.
	Henceforth, we refer to this graph $\cG$ as the {\em input-output graph}.
	
	\item Determine if there is a unique perfect matching in $\cG$. If the matching $\sigma$ is unique, return $\sigma$. If the matching is not unique, return {\sc Failure}.  
\end{enumerate}

Here, we discuss the running time of JEDI. 
For general codebooks, Step 1 can be implemented in $O(M^2)$ time.
Next,  we let $\cG$ be the input-output graph constructed in Step 1 and $E$ to be the number of edges in $\cG$.
Before we analyze Step 2, we first state some properties of $\cG$.
For each codeword $\vx\in \cC$, let $Y(\vx)$ be its corresponding channel output and we have that $Y(\vx)$ matches $\vx$. Therefore, the set of edges $\{(\vx,Y(\vx)) :  \vx\in \cC\}$ is a perfect matching of $\cG$.
Hence, we have two sub-tasks in Step 2: finding a perfect matching (since it exists) and determining if the matching is unique.
For the first sub-task, we can use the Hopcraft-Karp algorithm \cite{HopcroftKarp1973} to find a perfect matching in $O(E\sqrt{M})$ time. 
For the second sub-task, we can follow the methods described in Fukada \cite{Fukada1994} and Hoang \etal{} \cite{Hoang2006}, and then determine if another perfect matching exists in $O(M+E)=O(E)$ time (since $M\le E$). 
Hence, combining the analysis of both sub-tasks, we have that Step 2 can be implemented in $O(E\sqrt{M})=O(M^{2.5})$ time. 
Therefore, this simple analysis shows that JEDI runs in $O(M^{2.5})$ time.

Nevertheless, the complexity of JEDI can be further reduced. 
We do so by {\em improving the running time of Step 2}.
Crucially, we exploit the fact that $\cG$ contains a perfect matching.
Now, if we are able to determine early if there is more than one perfect matching,  we need not continue to find a perfect matching.
To do so, we modify the classic peeling decoders used in graph-based codes \cite{Zyablov1976}. 
Intuitively, we search for degree-one nodes in the graph $\cG$. For any such node $u$, the edge $uv$ incident to $u$ necessarily belongs to a perfect matching and hence, we add it to the matching.
We then remove both nodes $u$ and $v$, and all other edges incident to $v$ and repeat the search for degree-one nodes. 
We have two scenarios. 
In the first scenario, we remove all nodes from $\cG$ and end up with a perfect matching.
In the second scenario, all remaining nodes have degree at least two and it can be shown that $\cG$ contains at least two perfect matchings (see Section~\ref{sec:pma}). 
Thus, in this case, we can terminate our search earlier.
A formal description is given below.

\vspace{2mm}
\noindent{\bf Peeling Matching Algorithm (PMA)}.

\noindent{\sc Input}: A bipartite graph $\cG$ (with $M$ left and $M$ right vertices) that contains at least one perfect matching.\\[1mm]
\noindent{\sc Output}: A perfect matching $\cM$ of $\cG$ if it is unique. Otherwise, {\sc Failure} is declared. 
\begin{enumerate}[(1)]
	\item Initialize $\cM$ to the empty set.
	\item Find a node $u$ in $\cG$ with degree one. Here, $u$ may be a left {\em or} right node. If there is no such node, go to Step 6.
	\item Let $uv$ be the unique edge incident to $u$ and add $uv$ to $\cM$.
	\item Remove nodes $u$ and $v$ and all edges incident to $v$.
	\item Repeat Step 2. 
	\item If $\cM$ is a perfect matching, return $\cM$. Otherwise, $|\cM|<M$ and we declare {\sc Failure}.
\end{enumerate}

\begin{example}\label{exa:JEDI}
	Consider the simplex code of length seven. More concretely, we consider the linear code with $M=8$ codewords generated by the matrix
	$\begin{bmatrix}
		1 & 0 & 0 & 0 & 1 & 1 & 1\\
		0 & 1 & 0 & 1 & 0 & 1 & 1\\
		0 & 0 & 1 & 1 & 1 & 0 & 1\\
	\end{bmatrix}$.
	\begin{enumerate}[(a)]
		\item Suppose the channel outputs are: 
		\[\begin{array}{cccc}
		00?????, & 001????, & ??????0, & ?0?0?1?, \\
		11????0, & ????001, & 0??????, & ????110.
		\end{array}\]
		Then the bipartite graph $\cG$ constructed in Step 1 of JEDI is given below. 
		Highlighted in {\color{blue}blue} is the unique bipartite matching $\cM$ in $\cG$.
		\begin{center}	
		\begin{tikzpicture}[scale = 0.8]
		\tikzset{vertex/.style = {circle,fill, inner sep = 2pt}}
		\tikzset{match/.style = {- ,blue,thick}}
		\node[vertex] (x1) at  (1,0) {};    \node at (-0.3,0)    {\tt 0000000 }; 
		\node[vertex] (x2) at  (1,-0.5) {}; \node at (-0.3,-0.5) {\tt 1000111 }; 
		\node[vertex] (x3) at  (1,-1) {};   \node at (-0.3,-1)   {\tt 0101011 }; 
		\node[vertex] (x4) at  (1,-1.5) {}; \node at (-0.3,-1.5) {\tt 0011101 }; 
		\node[vertex] (x5) at  (1,-2) {};   \node at (-0.3,-2)   {\tt 1101100 }; 
		\node[vertex] (x6) at  (1,-2.5) {}; \node at (-0.3,-2.5) {\tt 1011010 }; 
		\node[vertex] (x7) at  (1,-3) {};   \node at (-0.3,-3)   {\tt 0110110 }; 
		\node[vertex] (x8) at  (1,-3.5) {}; \node at (-0.3,-3.5) {\tt 1110001 }; 
		
		\node[vertex] (y1) at  (3,0)    {}; \node at (4.3,0)    {\tt 00????? }; 
		\node[vertex] (y2) at  (3,-0.5) {}; \node at (4.3,-0.5) {\tt 001???? };
		\node[vertex] (y3) at  (3,-1)   {}; \node at (4.3,-1)   {\tt ??????0 };
		\node[vertex] (y4) at  (3,-1.5) {}; \node at (4.3,-1.5) {\tt ?0?0?1? };
		\node[vertex] (y5) at  (3,-2)   {}; \node at (4.3,-2)   {\tt 11????0 };
		\node[vertex] (y6) at  (3,-2.5) {}; \node at (4.3,-2.5) {\tt ????001 };
		\node[vertex] (y7) at  (3,-3)   {}; \node at (4.3,-3)   {\tt 0?????? };
		\node[vertex] (y8) at  (3,-3.5) {}; \node at (4.3,-3.5) {\tt ????110 };
		
		\draw[match] (y1) -- (x1); \draw(y1) -- (x4);
		\draw[match] (y2) -- (x4);
		\draw (y3) -- (x7); \draw (y3) -- (x1); \draw (y3) -- (x5); \draw[match] (y3) -- (x6);
		\draw[match] (y4) -- (x2);
		\draw[match] (y5) -- (x5); 
		\draw[match]  (y6) -- (x8);
		\draw (y7) -- (x1); \draw[match] (y7) -- (x3); \draw (y7) -- (x4); \draw (y7) -- (x7);
		\draw[match] (y8) -- (x7);
		\end{tikzpicture}
		\end{center}
	Here, we list the edges {\em in the order} they are added to $\cM$ according to PMA.
	\[
	\begin{array}{cccc}
	(1000111, ?0?0?1?), &  (0101011, 0??????), &
	(1011010, ??????0), & (1110001, ????001), \\
	(0011101, 001????), &  (1101100, 11????0), &
	(0110110, ????110), &  (0000000, 00?????).
	\end{array}\]  
	
	\item Suppose the channel outputs are: 
	\[\begin{array}{cccc}
	0000000, & ?0??1?1, & 0110110, & ?0??1?1, \\
	1101100, & 1110001, & 0110110, & 1011010.
	\end{array}\]
	Then the bipartite graph $\cG$ constructed is given below. 
	\begin{center}	
		\begin{tikzpicture}[scale = 0.8]
		\tikzset{vertex/.style = {circle,fill, inner sep = 2pt}}
		\tikzset{cycle/.style = {- ,red,very thick}}
		\node[vertex] (x1) at  (1,0) {};  	\node at (-0.3,0)  {\tt 0000000 };
		\node[vertex] (x2) at  (1,-0.5) {}; \node at (-0.3,-0.5) {\tt 1000111 };
		\node[vertex] (x3) at  (1,-1) {}; 	\node at (-0.3,-1) {\tt 0101011 };
		\node[vertex] (x4) at  (1,-1.5) {}; \node at (-0.3,-1.5) {\tt 0011101 };
		\node[vertex] (x5) at  (1,-2) {}; 	\node at (-0.3,-2) {\tt 1101100 };
		\node[vertex] (x6) at  (1,-2.5) {}; \node at (-0.3,-2.5) {\tt 1011010 };
		\node[vertex] (x7) at  (1,-3) {}; 	\node at (-0.3,-3) {\tt 0110110 };
		\node[vertex] (x8) at  (1,-3.5) {}; \node at (-0.3,-3.5) {\tt 1110001 };
		
		\node[vertex] (y1) at  (3,0)    {}; \node at (4.3,0)    {\tt 0000000 };
		\node[vertex] (y2) at  (3,-0.5) {}; \node at (4.3,-0.5) {\tt ?0??1?1 };
		\node[vertex] (y3) at  (3,-1)   {}; \node at (4.3,-1)   {\tt 0110110 };
		\node[vertex] (y4) at  (3,-1.5) {}; \node at (4.3,-1.5) {\tt ?0??1?1 };
		\node[vertex] (y5) at  (3,-2)   {}; \node at (4.3,-2)   {\tt 1101100 };
		\node[vertex] (y6) at  (3,-2.5) {}; \node at (4.3,-2.5) {\tt 1110001 };
		\node[vertex] (y7) at  (3,-3)   {}; \node at (4.3,-3)   {\tt 0110110 };
		\node[vertex] (y8) at  (3,-3.5) {}; \node at (4.3,-3.5) {\tt 1011010 };
		
		\draw (y1) -- (x1);
		\draw[cycle] (x2) -- (y4);\draw[cycle] (y4) -- (x4); \draw[cycle] (x4) -- (y2); \draw[cycle] (y2) -- (x2);
		\draw (x3) -- (y7);
		\draw (x5) -- (y5);
		\draw (x6) -- (y8);
		\draw (x7) -- (y3);
		\draw (x8) -- (y6);
		
		\end{tikzpicture}
	\end{center}
		Highlighted in {\color{red}red} is the edges remaining  after all degree-one nodes are removed. Hence, PMA declares {\sc Failure}. \qedhere
	\end{enumerate}
\end{example}

The following lemma on the correctness of PMA and its running time can be proved using the notion of stopping sets in peeling decoders \cite{Zyablov1976}. 
For completeness, we provide a detailed proof in Section~\ref{sec:pma}.

\begin{lemma}\label{lem:pma}
	Let $\cG$ be a balanced bipartite graph of order $M$ with $E$ edges. Suppose that $\cG$ contains at least one perfect matching.
If the perfect matching is unique, then the output $\cM$ of PMA is the perfect matching.
Otherwise, PMA declares {\sc Failure}.
Furthermore, PMA terminates in $O(E)$ time.
\end{lemma}

%
%
%
%
%

Therefore, since $E\le M^2$, we have that JEDI terminates in $O(M^2)$ time. 
Now, this running time analysis assumes the worst case where $\cG$ is a complete bipartite graph.
By design, the codebook $\cC$ is chosen such that most erasure patterns are correctable with high probability. 
In other words, each right node or channel output is expected to match with exactly one left node or codeword, and so, 
we expect the graph $\cG$ to be sparse.

It turns out that the expected number of edges in $\cG$ is given by the distance enumerator (see for example \cite{MS77}). 
Specifically, given a code $\cC$ of length $n$, we define $B_i$ to be the number of pairs of (not necessarily distinct) codewords of distance $d$. So, we have that $B_0=M$ and $\sum_{i=0}^n B_i=M^2$.
We then define the {\em distance enumerator of code $\cC$} to be polynomial $B(z) = \sum_{i=0}^n B_i z^i$. 

\begin{lemma}\label{lem:wt-enum}
	Consider a BEC with erasure probability $p$.
	If the distance enumerator of code $\cC$ is $B(z)$,
	then expected number of edges in $\cG$ constructed in JEDI is given by $B(p)$.
\end{lemma}

\begin{proof}
 Consider two codewords $\vx$ and $\vz$ and let $\tilde{\vx}$ be the channel output of $\vx$. 
 We first compute the probability that there is an edge between the nodes $\vz$ and $\tilde{\vx}$ in JEDI.
 Suppose the distance between $\vx$ and $\vz$ is $d$. 
 Then there is an edge between $\vz$ and $\tilde{\vx}$ if the $d$ bits where $\vx$ and $\vz$ differ are erased. 
 In other words, there is an edge between $\vz$ and $\tilde{\vx}$ with probability $p^d$. By linearity of expectation, we have that the expected number of edges is $\sum_{\vx,\vz\in\cC} z^{d_H(\vx,\vz)}= B(z)$. Here, $d_H(\vx,\vz)$ denotes the Hamming distance of $\vx$ and $\vz$.
\end{proof}

\begin{remark}
Alternatively, given $\cC=\{\vx_1,\vx_2,\ldots,\vx_M\}$, we can define an $(M\times M)$-matrix $\TBEC$ whose $(i,j)$-entry is given by $p^{d_H(\vx_i,\vx_j)}$. 
Then the sum of all $M^2$ entries in $\TBEC$ yields the distance enumerator for $\cC$.
While this method of determining the expected number of edges is computationally equivalent, it turns out the permanent of the matrix $\TBEC$ can be used to estimate the error probability of JEDI. We describe this in detail in Section~\ref{sec:prob}.
\end{remark}

Next, we consider a family of block codes and we determine sufficient conditions such the expected running of JEDI is linear in the block sizes $M_n$.

\begin{theorem}\label{thm:jedi-average}
	Consider a BEC with erasure probability $p$.
	Let $\cC$ be a linear code of size $M$ with minimum distance $d$.
	Then the expected number of edges of $\cG$ is at most $M+\binom{M}{2}p^d$. Hence, the expected running time of JEDI is $O\left(M+\binom{M}{2}p^d +Mn^3\right)$.
	
	Furthermore, suppose that $\{\cC_n\}_{n\ge 0}$ is a family of linear codes such that $\cC_n$ is a linear code of size $M_n$ with minimum distance $d_n$.
	If $\binom{M_n}{2}p^{d_n} = o(1)$, then the expected running time of JEDI tends to $O(M_n n^3)$.	
\end{theorem} 

\begin{proof}
First, we consider a code $\cC$ with minimum distance $d$. 
Then we have that $B_1=B_2=\cdots=B_{d-1}=0$. 
Since $p^{i}\le p^d$ for $i\ge d$, 
we have that $B(p) = B_0 + \sum_{i=d}^n B_ip^i \le M+\binom{M}{2}p^d$. Therefore, the expected number of edges in $\cG$ is at most $M+\binom{M}{2}p^d$. 
Hence, Step 2 runs in $O(M+\binom{M}{2}p^d)$ time.

So, it remains to improve the running time of Step 1, i.e. the time to construct the input-output graph $\cG$.
Since $\cC$ is a linear code of length $n$,
then for any channel output $\vy$, we can find all $\vx$ such that $\vy$ matches $\vx$ in $O(n^3)$ time by matrix inversion.
Therefore, $\cG$ can be constructed in $O(Mn^3)$ time.

The asymptotic analysis for $\{\cC_n\}_{n\ge 0}$ follows from the preceding argument.
\end{proof}

\subsection{Reed-Muller Codes}

In this subsection, we apply Theorem~\ref{thm:jedi-average} to the ubiquitous class of Reed-Muller codes \cite{Reed1954} and 
derive the expected running time of JEDI on this class of linear codes.

\begin{theorem}\label{thm:rm-jedi}
	Fix $r\ge 1$ and consider the family of $r$-th order Reed-Muller codes.
	Then for any $\epsilon>0$, 
	the expected running time of JEDI is
	\begin{equation*}
	\begin{cases}
	O(M^{1+\epsilon}) & \text{when }r\ge 2,\\
	O(M^{2+\epsilon}) & \text{when }r = 1.
	\end{cases}
	\end{equation*}
\end{theorem}

\begin{proof}
	Consider $1\le r\le m$. Recall that the Reed-Muller code ${\cal RM}(r,m)$ has the following parameters:
	$n = 2^m$, $\log M = \sum_{i=0}^r \binom{m}{i}$ and $d=2^{m-r}$.
	
	
	First, we demonstrate the expected running time tends to $O(Mn^3)$.
	Following Theorem~\ref{thm:jedi-average}, it remains to show that $\binom{M}{2}p^d$ tends to zero as $m\to\infty$, or equivalently,
	\begin{equation}\label{claim}
		\log \binom{M}{2} + d \log p\to -\infty.
	\end{equation}
	
	Now, for ${\cal RM}(r,m)$, we have that $\log M \le (r+1)m^r$.
	Since $d = 2^{m-r}$, the left hand side of \eqref{claim} is upper bounded by $2(r+1)m^r + (2^{-r}\log p)2^m$.
	When $p<1$, this upper bound tends to $-\infty$ and 
	hence, we have that $\binom{M}{2}p^d$ tends to zero, as required.
	This means that the input-output graph $\cG$ has expected number of edges at most $M$.
	
	Next, we bound the time required to construct $\cG$. Let $\epsilon$ be a positive constant.
	\begin{itemize}
		\item 	When $r\ge 2$, we use matrix inversion to construct $\cG$ as in the proof of Theorem~\ref{thm:jedi-average} and 
		we claim that $n^3=O(M^{\epsilon})$.
		Indeed, $M^\epsilon = 2^{\epsilon(\sum_{i=0}^r \binom{m}{i})}\ge 2^{\gamma m^r}$ for some constant $\gamma$.
		On the other hand, $n^3 = 2^{3m}$. 
		Since $3m = o\left(\gamma m^r\right)$, 
		we have that $2^{3m}=O\left(2^{\gamma m^r}\right)=O(M^\epsilon)$.
				
		\item 	When $r=1$, it is no longer true that $n^3=O(M^{\epsilon})$.
		Instead of using matrix inversion to construct $\cG$,
		we use the Fast Hadamard Transform (FHT) \cite{Green1966,Beery1986} to determine the edges.
		Specifically, for each channel output $\vy$, we can use FHT to find all $\vx$ that matches $\vy$ in $O(M\log M)$ time. 
		Therefore, $\cG$ can be constructed in $O(M^2\log M) =O(M^{2+\epsilon})$ time and we have the expected running time of JEDI as desired.
		\qedhere
	\end{itemize}
\end{proof}

To end this subsection, we comment that the derived running time is essentially optimal. 
As mentioned earlier, since we have to read all $M$ codewords of length $n$, a lower bound for the running time of any joint decoder is trivially $\Omega(Mn)$.
For Reed-Muller codes of order $r\ge 2$, we have $n=o(M)$ and the expected running time of $O(M^{1+\epsilon})$ is almost optimal.
When $r=1$, we note that $n=\Theta(M)$ and thus, $\Omega(Mn)=\Omega(M^2)$.
Again, the expected running time of $O(M^{2+\epsilon})$ is almost optimal.

%
%


\subsection{Correctness and Running Time of PMA}
\label{sec:pma}

For completeness, we provide a detailed proof of Lemma~\ref{lem:pma}.
Following Fukada \cite{Fukada1994} and Hoang \etal{} \cite{Hoang2006}, we define the notion of alternating cycle. 
Formally, consider a bipartite graph $\cG$ with a perfect matching $\cM$. A cycle $\cO$ in $\cG$ is {\em alternating} with respect to $\cM$ if the edges in $\cO$ alternate between in $\cM$ and not in $\cM$. We have the following lemma.

\begin{lemma}\label{lem:alternating}
	Let $\cM$ be a perfect matching in a balanced bipartite $\cG$. Then
	$\cM$ is the unique perfect matching in $\cG$ if and only if 
	there is no alternating cycle with respect to $\cM$.
\end{lemma}

\noindent{\em Correctness of PMA}. 
Applying Lemma~\ref{lem:alternating}, it suffices to show the following claim.

\begin{quote}
	Let $\cG$ be a balanced bipartite graph with a perfect matching $\cM$.
	If all the degrees of $\cG$ are at least two,
	then there is an alternating cycle with respect to $\cM$.
\end{quote}

Indeed, we construct an alternating cycle as follow.
Pick any left node $u_1$ in $\cG$. 
Since $\cM$ is a perfect matching, we can find a right node $v_1$ such that $u_1v_1$ belongs to $\cM$. Now, as the degree of $v_1$ is at least two, we can find $u_2$ such that $u_2v_1$ is an edge not belonging to $\cM$. 
We then repeat the process to find $v_2$ and $u_3$
such that $u_2v_2\in \cM$ and $u_3v_2\notin\cM$.
Since the degrees of all nodes are at least two, we are always able to find a left node $u_i$ and eventually, we have two left nodes that coincide and obtain an alternating cycle with respect to $\cM$.
\vspace{1mm}

\noindent{\em Running Time of PMA}. We briefly describe a data structure that implements PMA in time linear in the number of edges.
Recall that $\cG$ has $E$ edges and $2M$ nodes with $E\ge V$.

We maintain an {\em adjacency list} for the nodes. 
In other words, for each node $v$, we maintain a list $N(v)$ of nodes adjacent to $v$. Also, we maintain a queue $Q$ of degree-one nodes.

Whenever $Q$ is nonempty, we remove the first node $u$ and its neighbor $v$ and update the adjacency lists of the neighbors of $v$.
If any node becomes degree-one, we add it to the $Q$.
We continue this process until the queue is empty.
Since the number of updates to the adjacency lists is at most the number of edges, the running time of PMA is $O(M + E) = O(E)$.

\section{Joint Minimum-Distance Decoding Identifier}
\label{sec:jmdi}

We propose an efficient joint decoder for the binary symmetric channel (BSC). While our exposition assumes a BSC channel, we remark that the decoder can be modified to serve as a joint maximum likelihood decoder for other channels (see Section~\ref{sec:discussion}).  

As with Section~\ref{sec:jedi}, we reduce the problem of permutation recovery to that 
of finding a minimum-cost matching.
Specifically, consider a balanced bipartite graph $\cG$ of order $M$. 
In addition, we associate each edge $e$ in $\cG$ with a cost $c(e)$ and the cost of a matching $\cM$ is the sum of the costs of all edges in $\cM$.
Suppose that $\cG$ contains at least one perfect matching.
A perfect matching in $\cG$ is {\em minimum-cost} if its cost is at most the cost of any other perfect matching in $\cG$. When $\cG$ is a complete bipartite graph, the problem of finding a minimum-cost matching in $\cG$ is also known as the {\em assignment} problem and the Hungarian method finds a minimum-cost assignment in $O(M^3)$ time \cite{Kuhn1955}.

As with before, we use the codewords and channel outputs as the left and right nodes of a balanced bipartite graph $\cG$. 
Then for any codeword-output pair $(\vx,\vy)$, 
we connect them with an edge of cost $d_H(\vx,\vy)$.
Then the problem of finding a permutation $\sigma$ that maximizes the the probability 
$\prod_{i\in [M]} P(\vy_{\sigma(i)}| \vx_i)$ is equivalent to minimizing the cost of a perfect matching in $\cG$.
A formal description of the decoder is given below.


\noindent{\bf Joint Minimum-Distance Decoding Identifier (JMDI)}.

\noindent{\sc Input}: 
A codebook $\cC = \{\vx_1,\vx_2,\ldots, \vx_M\}$ of size $M$ and 
a set of $M$ channel outputs $\{\vy_1, \vy_2, \ldots , \vy_M\} \subseteq\{0,1\}^n$.\\[1mm]
\noindent{\sc Output}: A permutation $\sigma$ such that the quantity $\sum_{i\in [M]}d_H(\vx_i,\vy_{\sigma(i)})$ is minimized.
\begin{enumerate}[(1)]
	\item We draw a balanced bipartite graph $\cG$ of order $M$.
	Here, the $M$ codewords are the left nodes while the $M$ channel outputs are the right nodes.
	For $i,j \in [M]$, we draw an edge between $\vx_i$ and $\vy_j$ and set its cost to be $d_H(\vx_i,\vy_j)$. Again, we refer to this bipartite graph $\cG$ as the {\em input-output graph}.
	
	\item Find a minimum-cost matching in $\cG$.  
\end{enumerate}

We discuss the running time of JMDI.
In Step 1, we can compute the distance between all pairs of words in $O(M^2n)$ time. In Step 2, we can apply the Hungarian method \cite{Kuhn1955} and hence, we have the following theorem.

\begin{theorem}
	JMDI finds a permutation in $O(M^3+M^2n)$ time.
\end{theorem}

Unlike the input-output graph constructed for the BEC, the input-output graph obtained by the JMDI is necessarily complete. 
So, to improve the running time for the case of BSC channels, 
we relax our goal of finding the exact solution.
Instead, we approximate it by determining a minimum-cost matching in a sparse subgraph $\cH$ of $\cG$. 

Specifically, for this sparse subgraph $\cH$, we consider only edges whose cost is at most $R$ for some $R<n$. 
Then the degree of each right node/channel output $\vy$ is given by the number of codewords whose distance is at most $R$ from $\vy$. 
When $\cC$ is a code with minimum distance $d$ and $R\le (d-1)/2$, this number is one.
When $R> (d-1)/2$, this quantity is studied in the context of {\em list decoding}. 

Formally, a $\cC$ is $(R,L)$-list-decodable if for all $\vy\in\{0,1\}^n$, we have that $|\{\vx\in \cC: d_H(\vx,\vy)\le R\}|$ is at most $L$. 
Then we modify Step 1 of JMDI by only including edges with weight at most $R$. Let $\cH$ be the resulting bipartite graph and from the list-decoding property of $\cC$, we have that $\cH$ has at most $ML$ edges.
We then proceed as in Step 2 to find a minimum-cost matching and we call this method {\em joint list decoding identifier} (JLDI).
For sparse bipartite graphs with $V$ nodes and $E$ edges, a minimum-cost matching can be found in $O(V^2\log V+VE)$ time \cite{EdmondsKarp1972,Tomizawa1971} and hence, JLDI terminates in $O(M^2(\log M+L+n))$ time.

\begin{example}
	\label{exa:JLDI}
	Consider the linear code $\cC$ with $M=4$ codewords generated by the matrix
	$\begin{bmatrix}
	1 & 1 & 1 & 0 & 0 \\
	0 & 0 & 1 & 1 & 1 \\
	\end{bmatrix}$.
	
	Suppose the channel outputs are: 
	\[
	\vy_1=10000, \quad  \vy_2=11101, \quad \vy_3=00011, \quad \vy_4=10001.\]
	Below we present the bipartite graph $\cG$ constructed by JMDI. 
	To reduce clutter, we use a $4\times 4$-table whose $(i,j)$ entry is given by the cost of the edge $(\vx_i,\vy_j)$, i.e. $d_H(\vx_i,\vy_j)$.
	We refer to this table as the cost matrix of $\cG$.

	\begin{center}
		Cost Matrix of $\cG$:\, 
		\begin{tabular}{c|cccc}
			& $\vy_1$ & $\vy_2$ & $\vy_3$ & $\vy_4$ \\ \hline
	$\vx_1=00000$ & {\color{blue}1} & 4 & 2 & 2  \\
	$\vx_2=11100$ & 2 & {\color{blue}1} & 5 & 3  \\
	$\vx_3=00111$ & 4 & 3 & {\color{blue}1} & 3  \\
	$\vx_4=11011$ & 3 & 2 & 2 & {\color{blue}2}  \\
		\end{tabular}
	\end{center}
	
	 Highlighted in {\color{blue}blue} are the edges in a minimum-cost matching of $\cG$.
	 Here, the minimum-cost matching $\cM$ is given by $\{(\vx_i,\vy_i) : i\in [4]\}$.
	 
	 Now, we can verify that $\cC$ is a $(2,3)$-list-decodable code. Hence, if we apply JLDI with radius $R=2$, we obtain the bipartite graph $\cH$ whose cost-matrix is as follows.
	 
	 \begin{center}
	 	Cost Matrix of $\cH$:\,
	 	\begin{tabular}{c|cccc}
	 			& $\vy_1$ & $\vy_2$ & $\vy_3$ & $\vy_4$ \\ \hline
	 		$\vx_1=00000$ & {\color{blue}1} & -- & 2 & 2  \\
	 		$\vx_2=11100$ & 2 & {\color{blue}1} & -- & --  \\
	 		$\vx_3=00111$ & -- & -- & {\color{blue}1} & --  \\
	 		$\vx_4=11011$ & -- & 2 & 2 & {\color{blue}2}  \\
	 	\end{tabular}
	 \end{center}
 Indeed, we observe that the degree of $\vy_i$, $i\in [4]$, is at most three, corroborating the list-decoding property of $\cC$. In this case, we have that the minimum-cost matching of $\cH$ is also $\cM$. 
 \qedhere
	 
%
%
%
	
\end{example}

However, a minimum-cost matching in $\cH$ may not be a minimum-cost matching in $\cG$. 
In other words, JLDI may not return the same output as JMDI. 
Nevertheless, such cases occur with small probability 
and we provide upper bounds on the probabilities of such events. 

\begin{theorem}\label{thm:JLDI}
	Consider a BSC channel with crossover probability $p$.
	Let $\cC$ be an $(R,L)$-list-decodable code of length $n$, size $M$ and $R>pn$. Set $\gamma = R/(pn) - 1$.
	Then JLDI terminates in $O(M^2(\log M+L+n))$ time.
	Furthermore,
	\begin{equation}\label{X}
	\prob{\text{JLDI correctly finds }\sigma~|~\text{JMDI correctly finds }\sigma} \ge \left(1-\exp(-\gamma^2pn/3)\right)^M\,.
	\end{equation}
	
\end{theorem}

\begin{proof}
	The running time analysis of JLDI is described in the preceding paragraphs.
	
	To derive the probability estimates, 
	we consider the random variable $Z_i$ that measures the number of errors in the output of codeword $\vx_i$, $i\in [M]$. 
	In other words, $Z_i = d_H(\vx_i, \vy_{\sigma(i)})$.
	Let $\cG$ and $\cH$ be the bipartite graphs constructed in JMDI and JLDI, respectively.
	To simplify our arguments, we assume that the minimum-cost matchings in both graphs are unique and let them be $\cM_\cG$ and $\cM_\cH$.
	
	First, we argue that if $Z_i\le R$ for all $i\in[M]$ and JMDI is correct, then JLDI is necessarily correct. Since JMDI is correct, we have that the matching $\cM_\cG$ corresponds to $\sigma$.
	Also, for all $i$, since $Z_i\le R$, 
	we have the edge $(\vx_i, \vy_{\sigma(i)})$ has cost at most $R$. Therefore, the matching $\cM_\cG$ is still present in the graph $\cH$ and so, the minimum-cost matching $\cM_\cH$ is identical to $\cM_\cG$ and corresponds to $\sigma$. Hence, JLDI is correct.
	
	Next, we lower bound the probability that all values of $Z_i$ is at most $R$. 
	Fix $i\in [M]$. Using Chernoff's bound (see for example, \cite{Mitzenmacher2005}), we have that $P(Z_i\ge R)=P(Z_i\ge (1+\gamma)pn)\le \exp(-\gamma^2pn/3)$ and so, $P(Z_i\le R)\ge 1-\exp(-\gamma^2pn/3)$. Since the values of $Z_i$ are independent of each other, the lower bound \eqref{X} follows.
\end{proof}

\subsection{Reed-Muller Codes}

Similar to before, we verify that the class of Reed-Muller codes satisfy the conditions of Theorem~\ref{thm:JLDI}. 

\begin{theorem}\label{thm:rm-jldi}
	Fix $r\ge 1$ and consider the family of $r$-th order Reed-Muller codes.
	Consider further a BSC with crossover probability $p<1/2^r$.
	
	For any $\epsilon>0$, JLDI runs in $O(M^{2+\epsilon})$ time and 
	event \eqref{X} occurs with probability approaching one.  	
\end{theorem}

To demonstrate the result, we apply the following result on the list-decoding capabilities of Reed-Muller codes.
\begin{theorem}[Bhomick and Lovett~\cite{Bhomick2018}]\label{thm:lovett}
	Fix $r$ and $\alpha$ and set $R=(1/2^r-\alpha)n$.
	Then there is a constant $L_{r,\alpha}$ (dependent only on $r$ and $\alpha$) such that
	the Reed-Muller code ${\cal RM}(r,m)$ is $(R,L_{r,\alpha})$-list-decodable for all $m$.
\end{theorem}

 \begin{proof}[Proof of Theorem~\ref{thm:rm-jldi}]
 	Choose some small $\alpha$ so that $p<1/2^r-\alpha$ and set $R=(1/2^r-\alpha)n$. 
 	Then Theorem~\ref{thm:lovett} states that the list size is upper bounded by a constant independent of $n$ and $M$. Hence, applying Theorem~\ref{thm:JLDI}, we have the running time is $O(M^2(\log M+L+n))=O(M^2(\log M+n))$.
 	
 	When $r\ge 2$, we have that $n=O(M^\epsilon)$ as in the proof of Theorem~\ref{thm:rm-jedi} and hence, the running time is $O(M^{2+\epsilon})$. 
 	When $r=1$, as before, we use FHT to compute the costs of the edges.
 	Specifically, for each channel output $\vy$, 
 	we use FHT to compute $d_H(\vx,\vy)$ for all codewords $\vx$ in $O(M\log M)$ time. 
 	Therefore, $\cG$ can be constructed in $O(M^2\log M) =O(M^{2+\epsilon})$ time.
 	 
 	Next, Theorem~\ref{thm:JLDI} also states that 
 	event \eqref{X} occurs with probability at least  $\theta_n\triangleq\left(1-\exp(-\gamma^2pn/3)\right)^M$ for some constant $\gamma$. 
 	Recall that $M$ is the code size $2^{\sum_{i=0}^r \binom{m}{i}}\le 2^{(r+1)m^r}  = O(2^{\lambda n})$ for all any constant $\lambda$. So, we can choose $\lambda$ small enough so that $2^\lambda < \exp(\gamma^2p/3)$. Therefore, $M\exp(-\gamma^2pn/3)$ approaches zero and we have 
 	\begin{align*}
 	\lim_{n\to\infty} \theta_n 
 	& = \lim_{n\to\infty} (1-\exp(-\gamma^2pn/3))^M\\
 	& = \lim_{n\to\infty} \left(\left(1-\frac{1}{\exp(\gamma^2pn/3)}\right)^{\exp(\gamma^2pn/3)}\right)^{\frac{M}{\exp(\gamma^2pn/3)}}\\
 	& = \lim_{n\to\infty} e^{-M\exp(-\gamma^2pn/3)}=1, \text{as desired}
 	\qedhere
 	\end{align*}
 \end{proof}


\section{Probability of Erroneous Identification}
\label{sec:prob}

In this section, we provide probability estimates for the event where the joint decoders (described in Sections~\ref{sec:jedi} and~\ref{sec:jmdi}) fail to identify the codewords / bees. 
Specifically, throughout this section, we let $\cC$ denote a code with $M$ length-$n$ codewords $\vx_1,\vx_2,\ldots, \vx_M$.
As before, we send these $M$ codewords through a noisy channel and 
let $\{\vy_1, \vy_2,\ldots, \vy_M\}$ be the corresponding set of outputs.
Then there exists a permutation $\sigma\in \bbS_M$ such that $\vy_{\sigma(i)}$ is the noisy output of $\vx_i$ for $i\in[M]$.

Let $\sigma^*$ be the permutation / matching returned by the joint decoder and our task to estimate the probability that $\sigma^*\ne \sigma$, or, equivalently, the probability that $\sigma^*\sigma^{-1}\ne \id$. Here, $\sigma^{-1}$ denotes the inverse permutation of $\sigma$ while $\id$ denotes the identity permutation.
Without loss of generality, we assume that $\sigma=\id$ and we consider the following error events.

\begin{itemize}
	\item {\em Binary erasure channels (BEC)}. 
	Recall that JEDI returns the correct perfect matching / permutation $\id$ if and only if $\id$ is the unique perfect matching in the input-output graph $\cG$ constructed in Step 1. In other words, if the graph contains another matching $\sigma\in \bbS_M^*\triangleq \bbS_M\setminus\{\id\}$, JEDI declares failure. 
	Hence, we are interested in the event where $\sigma$ appears as a matching in the graph $\cG$ and we denote this event by $I(\sigma)$. 
	Note that $\prob{I(\id)}=1$.
	\item {\em Binary symmetric channels (BSC)}.
	Recall that JMDI returns a minimum-cost perfect matching from the weighted input-output graph constructed in Step 1. Hence, if the cost of the identity matching $\id$ is smaller than all other matchings, JMDI necessarily returns $\id$. 
	Hence, for some non-identity permutation $\sigma\in\bbS_M^*$, we study the event where the cost of matching $\sigma$ is at most the cost of the identity matching and for convenience, we also denote this event by $I(\sigma)$.  
	As with the binary erasure channel, we have that $\prob{I(\id)}=1$.
\end{itemize}

Therefore, our task is to estimate the quantity $\perror{\cC}  \triangleq \prob{\bigcup_{\sigma\in\bbS_M^*} I(\sigma)}$. 
Now, by the union bound, we have that $\perror{\cC} \le \sum_{\sigma\in\bbS_M^*} \prob{I(\sigma)}$. 
If we further define $U \triangleq \sum_{\sigma\in\bbS_M} \prob{I(\sigma)}$, it is straightforward to see that 
\begin{equation}\label{upper}
\perror{\cC}\le U-1.
\end{equation}

On the other hand, by Bonferroni inequalities / principles of inclusion-exclusion, we have that 
\[\perror{\cC}\ge \sum_{\sigma\in\bbS_M^*} \prob{I(\sigma)} - \sum_{\substack{\sigma\ne \tau,\\
		\sigma, \tau\in\bbS_M^*}} \prob{I(\sigma)\wedge I(\tau)}\,.\]
Proceeding similar as before, we define $V\triangleq \sum_{\sigma,\tau \in\bbS_M} \prob{I(\sigma)\wedge I(\tau)}$ and we have 
{\footnotesize 
\begin{align*}
	V & = 
	\sum_{\substack{\sigma\ne \tau,\\
			\sigma, \tau\in\bbS_M^*}} \prob{I(\sigma)\wedge I(\tau)}
	+ 2\sum_{\sigma\in \bbS_M} \prob{I(\sigma)\wedge I(\id)}
	+ \sum_{\sigma\in \bbS_M} \prob{I(\sigma)\wedge I(\sigma)}
- 2	\prob{I(\id)\wedge I(\id)}\\
	&= \sum_{\substack{\sigma\ne \tau,\\
			\sigma, \tau\in\bbS_M^*}} \prob{I(\sigma)\wedge I(\tau)}
	+ 2\sum_{\sigma\in \bbS_M} \prob{I(\sigma)}
	+ \sum_{\sigma\in \bbS_M} \prob{I(\sigma)}
	- 2	\\
	& = \sum_{\substack{\sigma\ne \tau,\\
			\sigma, \tau\in\bbS_M^*}} \prob{I(\sigma)\wedge I(\tau)} +3U-2.
\end{align*}
}

Therefore, together with the fact that $\sum_{\sigma\in\bbS_M^*} \prob{I(\sigma)}=U-1$, we have that
\begin{equation}\label{lower}
	\perror{\cC} \ge (U-1) - (V-3U+2) = 4U - V - 3\,.
\end{equation}

Hence, following \eqref{upper} and \eqref{lower}, to provide the required estimates on $\perror{\cC}$, it suffices to determine $U$ and $V$. 
However, $U$ and $V$ involve $M!$ and $(M!)^2$ summands, respectively, and 
in fact, we show later that $U$ can be expressed as a permanent function of a certain $M\times M$ matrix $\vT$.
Unfortunately, determining the permanent of a general matrix is computationally intractable \cite{Valiant.1979} and
the state-of-the-art methods of computing permanents, due to Nijenhuis-Wilf \cite{NW.1978} and Glynn \cite{Glynn.2010}, have running time $O(M2^{M-1})$.
In a recent work \cite{Kiah.2021}, we studied methods of computing permanents on trellises. 
While our trellis-based techniques did not significantly improve the running time of state-of-the-art methods for general matrices, 
we observed that, for structured matrices and permanent-like functions, we can borrow ideas from trellis theory \cite{Vardy.1998} to significantly reduce the running time.
We apply these techniques here to determine $V$.
Specifically, in Subsection~\ref{sec:computeV}, we show that $V$ can be computed in $O(M^{1.5} 4^{M})$ time.

\subsection{Estimating $U$}\label{sec:estimateU}

In this subsection, we provide estimates for $U$ using the following notion of permanents.

\begin{definition}
	Let $\vT$ be an $M\times M$-matrix whose $(i,j)$-th entry is $T_{ij}$.
	Then the {\em permanent} of $\vT$ is given by the value
	\begin{equation}\label{eq:perm}
		\perm(\vT)\triangleq \sum_{\sigma \in \bbS_M} \prod_{i\in[M]} T_{i\sigma(i)}\,.
	\end{equation}
\end{definition}

Given a code $\cC$ of size $M$ with codewords $\vx_1,\vx_2,\ldots, \vx_M$. 
Recall that $\TBEC$ denote an $(M\times M)$-matrix whose $(i,j)$-entry is given by $p^{d_H(\vx_i,\vx_j)}$. On the other hand, we let $\TBSCu$ and $\TBSCl$ denote two $(M\times M)$-matrices. The $(i,j)$-th entry of $\TBSCu$ is given by $\left(4p(1-p)\right)^{d_H(\vx_i,\vx_j)/2}$, while the $(i,j)$-th entry of $\TBSCl$ is given by $\left(p(1-p)\right)^{d_H(\vx_i,\vx_j)/2}$.
Then the next proposition states that the permanents of these matrices can be used to estimate $U$. 

\begin{proposition}\label{prop:permanent-bound}
	Let $\cC$ be a code and let $0<p<1/2$. Recall that $U \triangleq \sum_{\sigma\in\bbS_M} \prob{I(\sigma)}$.
	\begin{enumerate}[(i)]
		\item If we transmit the codebook $\cC$ through a ${\rm BEC}(p)$, then $U=\perm(\TBEC)$.
		\item If we transmit the codebook $\cC$ through a ${\rm BSC}(p)$, then $ \perm(\TBSCl)\le U\le \perm(\TBSCu)$.
	\end{enumerate}
\end{proposition}

\begin{proof}In both cases, following the definition of permanent \eqref{eq:perm}, it suffices to demonstrate a certain relationship between $\prob{I(\sigma)}$ and $\prod_{i\in[M]} T_{i\sigma(i)}$.
	
\begin{enumerate}[(i)]
\item First, we consider the binary erasure channel and fix some permutation $\sigma$. 
Then $I(\sigma)$ is the event that the input-output graph $\cG$ contains the perfect matching corresponding to $\sigma$. In other words, for all $i\in [M]$, there is an edge between the input $\vx_i$ and the output of $\vx_{\sigma(i)}$. 
As we argued in Lemma~\ref{lem:wt-enum}, this probability of given by $p^{d_H(\vx_i,\vx_{\sigma(i)})}$ and thus, $\prob{I(\sigma)} =  \prod_{i\in[M]} p^{d_H(\vx_i,\vx_{\sigma(i)})}$. Finally, the proposition then follows from the definition of $\TBEC$.

\item Next, we consider the binary symmetric channel and again, fix some permutation $\sigma$. Then the $I(\sigma)$ is the event where the cost of the matching $\sigma$ is at most the cost of the identity matching. In other words, we have the $(Mn)$-length word $\vx_{\sigma(1)}\vx_{\sigma(2)}\cdots \vx_{\sigma(M)}$ is closer to the output $\vy_1\vy_2\cdots \vy_M$ in terms of Hamming distance than the word $\vx_{1}\vx_{2}\cdots \vx_{M}$. Let $D$ denote the Hamming distance between $\vx_{\sigma(1)}\vx_{\sigma(2)}\cdots \vx_{\sigma(M)}$ and $\vx_{1}\vx_{2}\cdots \vx_{M}$. Then the probability of event $I_\sigma$ is given by $\sum_{j=\lceil D/2\rceil}^{D}\binom{D}{j}p^{j}(1-p)^{D-j}$\,. Now, Barg and Forney \cite[Section III]{BargForney.2002} provided the following estimates for the latter quantity. Specifically, they showed that  \begin{equation}\label{eq:BargForney}
(p(1-p))^{D/2} \le \sum_{j=\lceil D/2\rceil}^{D}\binom{D}{j}p^{j}(1-p)^{D-j} \le (4p(1-p))^{D/2}\,.
\end{equation}
Then, using $D=\sum_{i\in [M]} d_H(\vx_i,\vx_{\sigma(i)})$, we obtain the proposition.
\end{enumerate}
\end{proof}

Unfortunately, as mentioned earlier, determining the exact value of the permanent of a general matrix is computationally slow. 
Nevertheless, when the code has a certain minimum distance, we are able to use analytic combinatorics \cite{SedgewickFlajolet.2013} to provide an upper bound for $U$.

\begin{theorem}\label{thm:bound1}
	Suppose that $\cC$ is a code with minimum distance $d$.
	Set 
	\begin{equation}\label{eq:theta}
		\theta = 
		\begin{cases}
			p^d, & \text{if the channel is ${\rm BEC}(p)$},\\
			\left(4p(1-p)\right)^{d/2}, & \text{if the channel is ${\rm BSC}(p)$},\\
		\end{cases}	
	\end{equation}
	Then 
	\begin{equation}\label{eq:bound1}
		U\le M!\sum_{i=0}^M \frac{\theta^{M-i}(1-\theta)^i}{i!} \,.
	\end{equation}
	Therefore, $\perror{\cC}\le \left(M!\sum_{i=0}^M \frac{\theta^{M-i}(1-\theta)^i}{i!}\right)-1$\,. 
\end{theorem}

\begin{proof}
	Let $\vT$ be either $\TBEC$ or $\TBSCu$. 
	Then using Proposition~\ref{prop:permanent-bound}, it suffices to show that $\perm(\vT)$ is at most $\left(M!\sum_{i=0}^M \theta^{M-i}(1-\theta)^i/i!\right)$.
	
	Now, we consider $\vU$ whose entries are such that $\vU_{ii}=1$ and $\vU_{ij}=\theta$ if $i\ne j$. Then we observe that $\vT_{\ij}\le \vU_{ij}$ for all $i,j\in[M]$. Thus, $\perm(\vT)$ is at most $\perm(\vU)$.
	For each permutation $\sigma\in\bbS_M$, the product $\prod_{i\in[M]} \vU_{i\sigma(i)}$ is given by $\prod_{i\ne \sigma_i} \theta = \theta^{M-f(\sigma)}$, 
	where $\theta=p^d$ and $f(\sigma)=|\{i = \sigma_i:i\in [M]\}|$.
	Note that $f(\sigma)$ is also known as the number of fixed points of a permutation $\sigma$.
	
	Let $F_{j,M}$ be the number of permutations of length $M$ with exactly $j$ fixed points. It turns out the following exponential generating function is known (see \cite[Theorem 7.12]{SedgewickFlajolet.2013}).
	\begin{equation*}
		\sum_{j,M}  \frac{F_{j,M}}{M!} u^jz^M= \frac{e^{(u-1)z}}{1-z}
	\end{equation*}
	
	With this result, let us estimate $\sum_{\sigma\in {\bbS}_M}\theta^{M-f(\sigma)}$.
	
	\begin{align*}
		\sum_{M\ge 0} \frac{1}{M!}(\sum_{\sigma\in {\bbS}_M}\theta^{M-f(\sigma)})z^M 
		& = \sum_{M,j} \frac{F_{j,M}}{M!}\theta^{M-j}z^M \\
		& = \sum_{M,j} \frac{F_{j,M}}{M!}\theta^{-j}(\theta z)^M \\
		& = \frac{e^{(\theta^{-1}-1)\theta z}}{1-\theta z}
		= \frac{e^{(1-\theta)z}}{1-\theta z}
	\end{align*}
	
	Therefore, 
	$\sum_{\sigma\in {\bbS}_M}\theta^{M-f(\sigma)} 
	= M![z^M] \frac{e^{(1-\theta)z}}{1-\theta z}$.
	Now, since 
	$\frac{1}{1-\theta z}  = \sum_{i\ge 0}(\theta z)^i$ and
	$e^{(1-\theta) z}  = \sum_{i\ge 0}(1-\theta)^i z^i/i!\, $, we have that
	\[\sum_{\sigma\in {\bbS}_M}\theta^{M-f(\sigma)}  = M!\sum_{i=0}^M \frac{\theta^{M-i}(1-\theta)^i}{i!},\] 
	\noindent as required.
\end{proof}

\begin{figure}[!t]
	\scriptsize
	\begin{tikzpicture}[x = 5mm, y = 2.4mm]
		\tikzstyle{state}=[rectangle,fill=white,draw,line width=0.5mm, text width = 8mm, align=center]
		\tikzstyle{label}=[fill=white, inner sep=1pt]
		\node[state] at (0,0) (0)     {$\varnothing$};
		\node[state] at (30,0) (123)  {$123$};
		
		\node[state] at (10,40) (11) {$1,1$};
		\path[-] (0) edge node[label, pos = 0.8] {$T_{11}$} (11);
		\node[state] at (20,43) (1212) {$12,12$};
		\path[-] (11) edge node[label, pos = 0.8] {$T_{22}$} (1212);
		\path[-] (1212.east) edge node[label, pos = 0.2] {$T_{33}$} (123);
		\node[state] at (20,41) (1213) {$12,13$};
		\path[-] (11) edge node[label, pos = 0.8] {$T_{22}T_{32}$} (1213);
		\path[-] (1213.east) edge node[label, pos = 0.2] {$T_{32}T_{22}$} (123);
		\node[state] at (20,39) (1312) {$13,12$};
		\path[-] (11) edge node[label, pos = 0.8] {$T_{32}T_{22}$} (1312);
		\path[-] (1312.east) edge node[label, pos = 0.2] {$T_{22}T_{32}$} (123);
		\node[state] at (20,37) (1313) {$13,13$};
		\path[-] (11) edge node[label, pos = 0.8] {$T_{32}$} (1313);
		\path[-] (1313.east) edge node[label, pos = 0.2] {$T_{23}$} (123);
		\node[state] at (10,30) (12) {$1,2$};
		\path[-] (0) edge node[label, pos = 0.8] {$T_{11}T_{21}$} (12);
		\node[state] at (20,33) (1221) {$12,21$};
		\path[-] (12) edge node[label, pos = 0.8] {$T_{22}T_{12}$} (1221);
		\path[-] (1221.east) edge node[label, pos = 0.2] {$T_{33}$} (123);
		\node[state] at (20,31) (1223) {$12,23$};
		\path[-] (12) edge node[label, pos = 0.8] {$T_{22}T_{32}$} (1223);
		\path[-] (1223.east) edge node[label, pos = 0.2] {$T_{32}T_{12}$} (123);
		\node[state] at (20,29) (1321) {$13,21$};
		\path[-] (12) edge node[label, pos = 0.8] {$T_{32}T_{12}$} (1321);
		\path[-] (1321.east) edge node[label, pos = 0.2] {$T_{22}T_{32}$} (123);
		\node[state] at (20,27) (1323) {$13,23$};
		\path[-] (12) edge node[label, pos = 0.8] {$T_{32}$} (1323);
		\path[-] (1323.east) edge node[label, pos = 0.2] {$T_{22}T_{12}$} (123);
		\node[state] at (10,20) (13) {$1,3$};
		\path[-] (0) edge node[label, pos = 0.8] {$T_{11}T_{31}$} (13);
		\node[state] at (20,23) (1231) {$12,31$};
		\path[-] (13) edge node[label, pos = 0.8] {$T_{22}T_{12}$} (1231);
		\path[-] (1231.east) edge node[label, pos = 0.2] {$T_{32}T_{22}$} (123);
		\node[state] at (20,21) (1232) {$12,32$};
		\path[-] (13) edge node[label, pos = 0.8] {$T_{22}$} (1232);
		\path[-] (1232.east) edge node[label, pos = 0.2] {$T_{32}T_{12}$} (123);
		\node[state] at (20,19) (1331) {$13,31$};
		\path[-] (13) edge node[label, pos = 0.8] {$T_{32}T_{12}$} (1331);
		\path[-] (1331.east) edge node[label, pos = 0.2] {$T_{23}$} (123);
		\node[state] at (20,17) (1332) {$13,32$};
		\path[-] (13) edge node[label, pos = 0.8] {$T_{32}T_{22}$} (1332);
		\path[-] (1332.east) edge node[label, pos = 0.2] {$T_{22}T_{12}$} (123);
		\node[state] at (10,10) (21) {$2,1$};
		\path[-] (0) edge node[label, pos = 0.8] {$T_{21}T_{11}$} (21);
		\node[state] at (20,13) (2112) {$21,12$};
		\path[-] (21) edge node[label, pos = 0.8] {$T_{12}T_{22}$} (2112);
		\path[-] (2112.east) edge node[label, pos = 0.2] {$T_{33}$} (123);
		\node[state] at (20,11) (2113) {$21,13$};
		\path[-] (21) edge node[label, pos = 0.8] {$T_{12}T_{32}$} (2113);
		\path[-] (2113.east) edge node[label, pos = 0.2] {$T_{32}T_{22}$} (123);
		\node[state] at (20,9) (2312) {$23,12$};
		\path[-] (21) edge node[label, pos = 0.8] {$T_{32}T_{22}$} (2312);
		\path[-] (2312.east) edge node[label, pos = 0.2] {$T_{12}T_{32}$} (123);
		\node[state] at (20,7) (2313) {$23,13$};
		\path[-] (21) edge node[label, pos = 0.8] {$T_{32}$} (2313);
		\path[-] (2313.east) edge node[label, pos = 0.2] {$T_{12}T_{22}$} (123);
		\node[state] at (10,0) (22) {$2,2$};
		\path[-] (0) edge node[label, pos = 0.8] {$T_{21}$} (22);
		\node[state] at (20,3) (2121) {$21,21$};
		\path[-] (22) edge node[label, pos = 0.8] {$T_{12}$} (2121);
		\path[-] (2121.east) edge node[label, pos = 0.2] {$T_{33}$} (123);
		\node[state] at (20,1) (2123) {$21,23$};
		\path[-] (22) edge node[label, pos = 0.8] {$T_{12}T_{32}$} (2123);
		\path[-] (2123.east) edge node[label, pos = 0.2] {$T_{32}T_{12}$} (123);
		\node[state] at (20,-1) (2321) {$23,21$};
		\path[-] (22) edge node[label, pos = 0.8] {$T_{32}T_{12}$} (2321);
		\path[-] (2321.east) edge node[label, pos = 0.2] {$T_{12}T_{32}$} (123);
		\node[state] at (20,-3) (2323) {$23,23$};
		\path[-] (22) edge node[label, pos = 0.8] {$T_{32}$} (2323);
		\path[-] (2323.east) edge node[label, pos = 0.2] {$T_{13}$} (123);
		\node[state] at (10,-10) (23) {$2,3$};
		\path[-] (0) edge node[label, pos = 0.8] {$T_{21}T_{31}$} (23);
		\node[state] at (20,-7) (2131) {$21,31$};
		\path[-] (23) edge node[label, pos = 0.8] {$T_{12}$} (2131);
		\path[-] (2131.east) edge node[label, pos = 0.2] {$T_{32}T_{22}$} (123);
		\node[state] at (20,-9) (2132) {$21,32$};
		\path[-] (23) edge node[label, pos = 0.8] {$T_{12}T_{22}$} (2132);
		\path[-] (2132.east) edge node[label, pos = 0.2] {$T_{32}T_{12}$} (123);
		\node[state] at (20,-11) (2331) {$23,31$};
		\path[-] (23) edge node[label, pos = 0.8] {$T_{32}T_{12}$} (2331);
		\path[-] (2331.east) edge node[label, pos = 0.2] {$T_{12}T_{22}$} (123);
		\node[state] at (20,-13) (2332) {$23,32$};
		\path[-] (23) edge node[label, pos = 0.8] {$T_{32}T_{22}$} (2332);
		\path[-] (2332.east) edge node[label, pos = 0.2] {$T_{13}$} (123);
		\node[state] at (10,-20) (31) {$3,1$};
		\path[-] (0) edge node[label, pos = 0.8] {$T_{31}T_{11}$} (31);
		\node[state] at (20,-17) (3112) {$31,12$};
		\path[-] (31) edge node[label, pos = 0.8] {$T_{12}T_{22}$} (3112);
		\path[-] (3112.east) edge node[label, pos = 0.2] {$T_{22}T_{32}$} (123);
		\node[state] at (20,-19) (3113) {$31,13$};
		\path[-] (31) edge node[label, pos = 0.8] {$T_{12}T_{32}$} (3113);
		\path[-] (3113.east) edge node[label, pos = 0.2] {$T_{23}$} (123);
		\node[state] at (20,-21) (3212) {$32,12$};
		\path[-] (31) edge node[label, pos = 0.8] {$T_{22}$} (3212);
		\path[-] (3212.east) edge node[label, pos = 0.2] {$T_{12}T_{32}$} (123);
		\node[state] at (20,-23) (3213) {$32,13$};
		\path[-] (31) edge node[label, pos = 0.8] {$T_{22}T_{32}$} (3213);
		\path[-] (3213.east) edge node[label, pos = 0.2] {$T_{12}T_{22}$} (123);
		\node[state] at (10,-30) (32) {$3,2$};
		\path[-] (0) edge node[label, pos = 0.8] {$T_{31}T_{21}$} (32);
		\node[state] at (20,-27) (3121) {$31,21$};
		\path[-] (32) edge node[label, pos = 0.8] {$T_{12}$} (3121);
		\path[-] (3121.east) edge node[label, pos = 0.2] {$T_{22}T_{32}$} (123);
		\node[state] at (20,-29) (3123) {$31,23$};
		\path[-] (32) edge node[label, pos = 0.8] {$T_{12}T_{32}$} (3123);
		\path[-] (3123.east) edge node[label, pos = 0.2] {$T_{22}T_{12}$} (123);
		\node[state] at (20,-31) (3221) {$32,21$};
		\path[-] (32) edge node[label, pos = 0.8] {$T_{22}T_{12}$} (3221);
		\path[-] (3221.east) edge node[label, pos = 0.2] {$T_{12}T_{32}$} (123);
		\node[state] at (20,-33) (3223) {$32,23$};
		\path[-] (32) edge node[label, pos = 0.8] {$T_{22}T_{32}$} (3223);
		\path[-] (3223.east) edge node[label, pos = 0.2] {$T_{13}$} (123);
		\node[state] at (10,-40) (33) {$3,3$};
		\path[-] (0) edge node[label, pos = 0.8] {$T_{31}$} (33);
		\node[state] at (20,-37) (3131) {$31,31$};
		\path[-] (33) edge node[label, pos = 0.8] {$T_{12}$} (3131);
		\path[-] (3131.east) edge node[label, pos = 0.2] {$T_{23}$} (123);
		\node[state] at (20,-39) (3132) {$31,32$};
		\path[-] (33) edge node[label, pos = 0.8] {$T_{12}T_{22}$} (3132);
		\path[-] (3132.east) edge node[label, pos = 0.2] {$T_{22}T_{12}$} (123);
		\node[state] at (20,-41) (3231) {$32,31$};
		\path[-] (33) edge node[label, pos = 0.8] {$T_{22}T_{12}$} (3231);
		\path[-] (3231.east) edge node[label, pos = 0.2] {$T_{12}T_{22}$} (123);
		\node[state] at (20,-43) (3232) {$32,32$};
		\path[-] (33) edge node[label, pos = 0.8] {$T_{22}$} (3232);
		\path[-] (3232.east) edge node[label, pos = 0.2] {$T_{13}$} (123);		
	\end{tikzpicture}
\vfill
\caption{Non-minimal trellis that computes $\permsec(\vT)$ where $\vT$ is a $(3\times 3)$-matrix. Here, the trellis has 47 vertices and 81 edges.}
\label{fig:nonminimal}
\end{figure}

\subsection{Computing $V$} \label{sec:computeV}

Recall that $V\triangleq \sum_{\sigma,\tau \in\bbS_M} \prob{I(\sigma)\wedge I(\tau)}$.
As pointed out earlier, the quantity $V$ comprises $(M!)^2 = \Theta\left(M(M/e)^{2M} \right)$ summands and in this subsection, we borrow ideas from trellis theory~\cite{Vardy.1998} to reduce the running time to $O(M^{1.5}4^M)$. 

Before we discuss about trellis theory, we recall the definition of the matrices $\TBEC$, $\TBSCu$, and $\TBSCl$ and establish a relationship of the quantity $V$ with the following permanent-like matrix function.

\begin{definition}
	Let $\vT$ be an $M\times M$-matrix whose $(i,j)$-th entry is $T_{ij}$.
	For $i,k\in [M]$, we further define $\phi(T_{ij},T_{kj})=T_{ij}$ if $i= k$ and $\phi(T_{ij},T_{kj})=T_{ij}T_{kj}$ if $i\ne k$.
	Then the {\em second-order permanent} of $\vT$ is given by the value
	\begin{equation}\label{eq:perm}
		\permsec(\vT)\triangleq \sum_{\sigma,\tau \in \bbS_M} \prod_{j\in[M]} \phi(T_{\sigma(j)j},T_{\tau(j)j})\,.
	\end{equation}
\end{definition}

We have the following analogue of Proposition~\ref{prop:permanent-bound}.

\begin{proposition}\label{prop:permsec}
	Let $\cC$ be a code and let $0<p<1/2$. Recall that $V\triangleq \sum_{\sigma,\tau \in\bbS_M} \prob{I(\sigma)\wedge I(\tau)}$.
	\begin{enumerate}[(i)]
		\item If we transmit the codebook $\cC$ through a ${\rm BEC}(p)$, then $V=\permsec(\TBEC)$.
		\item If we transmit the codebook $\cC$ through a ${\rm BSC}(p)$, then $ \permsec(\TBSCl)\le V\le \permsec(\TBSCu)$.
	\end{enumerate}
\end{proposition}

\begin{proof}
The proof is similar to the proof of Proposition~\ref{prop:permanent-bound}. Hence, we only sketch the proof for the binary erasure channel.
Fix a pair of permutations $\sigma$ and $\tau$. 
Then $I(\sigma)\wedge I(\tau)$ is the event that input-output graph $\cG$ contains both perfect matchings $\sigma$ and $\tau$. Hence, for all $j\in [M]$, if $\sigma(j)\ne\tau(j)$, we have both edges: one edge between the input $\vx_j$ and the output of $\vx_{\sigma(j)}$, and another edge between the input $\vx_j$ and the output of $\vx_{\tau(j)}$. 
If $\sigma(j)=\tau(j)$, we have the edge between the input $\vx_j$ and the output of $\vx_{\sigma(j)}$. To conclude the proof, we then proceed as in the proof of Proposition 10(i).
\end{proof}

Therefore, to compute or estimate $V$, we determine the second-order permanents of $\TBEC$, $\TBSCu$, and $\TBSCl$. To do so, we adapt techniques from our recent work~\cite{Kiah.2021}, where we proposed a very different approach to exact permanent computation. 
In this subsection, we replicate and modify certain parts from~\cite{Kiah.2021} for our computation task.
Specifically, to compute the second-order permanents, we use a graph structure called {\em trellis}. 
The trellis was invented by Forney \cite{Forney.1967} over fifty years ago
to illustrate the Viterbi decoding algorithm \cite{Viterbi.1967} for
convolutional codes.
It has since been studied extensively by coding theorists; 
see \cite{Vardy.1998} for an excellent survey.

A {\em trellis} $\bbT=(\bbV, \bbE, \bbL)$ is an edge-labelled directed graph,
where $\bbV$ is the set of {\em vertices}, $\bbE$ is the~set~of ordered
pairs $(v,v')\in \bbV\times \bbV$, called {\em edges}, and $\bbL$ is the 
{\em edge-labelling function}. Specifically, $\bbL$~is~a~real-valued function that maps an edge $e\in \bbE$ to a real number. 
The defining property of a trellis is that the set $\bbV$ of vertices can
be partitioned into $\bbV_0,\bbV_1,\ldots, \bbV_M$ such that every edge
$(v,v')$ begins at $v\in \bbV_{j-1}$ and terminates at $v'\in \bbV_j$ for
some $j\in [M]$. In addition, the subsets $\bbV_0$ and $\bbV_M$
are singletons, containing two distinguished vertices, called the {\em root} and the {\em toor}, respectively.  

For each path $\vp$ defined by its edge sequence $e_1e_2\cdots e_t$, we associate the path $\vp$ with its {\em value}	$\val(\vp)\triangleq \prod_{j=1}^t \bbL(e_j)$.  Furthermore, we use $\bbP(\bbT)$ to denote the {\em multiset} of all paths from the root to toor and we are interested in the {\em value} of the trellis, defined by $\val(\bbT)\triangleq \sum_{\vp\in\bbP(\bbT)} \val(\vp)$. 
Then the celebrated Viterbi algorithm%
\footnote{
	We omit a detailed description of the Viterbi algorithm here and instead, refer the interested reader to the meticulous exposition in \cite{Vardy.1998}. 
	The original algorithm was introduced by Viterbi \cite{Viterbi.1967} in 1967 to perform maximum-likelihood decoding of convolutional codes. 
	In~\cite{Kiah.2021}, we described the Viterbi algorithm in the context of permanent computations.} is an application of the dynamic programming
method that computes  $\val(\vp)$ efficiently. 
Specifically, the Viterbi algorithm computes $\val(\bbT)$ using exactly $|\bbE|-\deg({\rm root})$ multiplications and exactly $|(\bbE|-|\bbV|+1)$ additions.

\begin{example}\label{exa:nonminimal}
Set $M=3$ and we consider a $3\times 3$-matrix $\vT$. 
Then we can use the trellis $\bbT$ in Figure~\ref{fig:nonminimal} to compute $\permsec(\vT)$.
In particular, $\val(\bbT) = \permsec(\vT)$.
We can readily check that $|\bbV_0|=|\bbV_3|=1$, $|\bbV_1|=9$ and $|\bbV_2|=36$ and so, $|\bbV|=47$ and $|\bbE|=81$.
Therefore, the Viterbi algorithm computes $\val(\bbT) = \permsec(\vT)$ using 72 multiplications and 35 additions.
\end{example}

Consider some trellis $\bbT$. 
Since the complexity of evaluating the value of $\bbT$ depends on its number of vertices and edges, one key objective in the study of trellises in coding theory is to find a ``smaller'' trellis $\bbT'$ so that $\bbP(\bbT')=\bbP(\bbT)$. 
Formally, we say that $\bbT^*=(\bbV^*,\bbE^*,\bbL^*)$ is a {\em minimal trellis} for $\bbT$ if $\bbP(\bbT^*)=\bbP(\bbT)$ and the following holds:
\begin{quote}
for all other trellises $\bbT'=(\bbV,\bbE,\bbL)$ such that $\bbP(\bbT')=\bbP(\bbT)$, 
		we have that $|\bbV^*_j|\le |\bbV_j|$ for all $j=0,1,\ldots, M$.
\end{quote}

There are examples of trellises that do not admit a minimal
trellis representation.  Nevertheless, if the $\bbT$ obeys certain properties,
we have that $\bbT$ admits a unique minimal trellis. Moreover, there is a simple {\em merging} procedure that finds this trellis \cite{Kschischang.1996, VardyK.1996}.
Here, we omit the details of the merging process. 
Instead, we simply describe the minimal trellis $\bbT^*_M$ that computes $\permsec(\vT)$ for some given matrix $\vT$.
 
\begin{figure}[!t]
	\begin{tikzpicture}[x = 5mm, y = 7mm]
		\small
		\tikzstyle{state}=[rectangle,fill=white,draw,line width=0.5mm, text width = 9mm, align=center]
		\tikzstyle{label}=[fill=white, inner sep=1pt]
		\node[state] at (0,0) (0) {$\varnothing$};
		\node[state] at (10,5) (1) {$1$};
		\node[state] at (10,3) (2) {$2$};
		\node[state] at (10,1) (3) {$3$};
		\node[state] at (10,-1) (1+2) {$1,2$};
		\node[state] at (10,-3) (1+3) {$1,3$};
		\node[state] at (10,-5) (2+3) {$2,3$};
		\node[state] at (20,5) (12) {$12$};
		\node[state] at (20,3) (13) {$13$};
		\node[state] at (20,1) (23) {$23$};
		\node[state] at (20,-1) (12+13) {$12,13$};
		\node[state] at (20,-3) (12+23) {$12,23$};
		\node[state] at (20,-5) (13+23) {$13,23$};
		\node[state] at (30,0) (123) {$123$};
		
		\draw[green] (0) -- (1);
		\draw[green] (0) -- (2);
		\draw[green] (0) -- (3);
		\draw[red] (0) -- (1+2);
		\draw[red] (0) -- (1+3);
		\draw[red] (0) -- (2+3);
		
		\draw[green] (1) -- (12);
		\draw[green] (1) -- (13);
		\draw[red]   (1) -- (12+13);
		\draw[green] (2) -- (12);
		\draw[green] (2) -- (23);
		\draw[red]   (2) -- (12+23);		
		\draw[green] (3) -- (13);
		\draw[green] (3) -- (23);
		\draw[red]   (3) -- (13+23);		
		
		\draw[blue]    (1+2) -- (12);	
		\draw[blue]    (1+2) -- (12+23);	
		\draw[blue]    (1+2) -- (12+13);	
		\draw[green]   (1+2) -- (13+23);	
		\draw[blue]    (1+3) -- (13);	
		\draw[blue]    (1+3) -- (12+13);	
		\draw[blue]    (1+3) -- (13+23);	
		\draw[green]   (1+3) -- (12+23);	
		\draw[blue]    (2+3) -- (23);	
		\draw[blue]    (2+3) -- (12+23);	
		\draw[blue]    (2+3) -- (13+23);	
		\draw[green]   (2+3) -- (12+13);			
		
		\draw[green] (12) -- (123);
		\draw[green] (13) -- (123);
		\draw[green] (23) -- (123);		
		\draw[blue] (12+13) -- (123);
		\draw[blue] (12+23) -- (123);
		\draw[blue] (13+23) -- (123);		
	\end{tikzpicture}
	\caption{Minimal trellis $\bbT^*_3$ that computes $\permsec(\vT)$ where $\vT$ is a $(3\times 3)$-matrix. To reduce clutter, edge labels have been removed. Instead, {\color{green} green}, {\color{blue} blue}, and {\color{red} red} edges denote edges with labels of the form ${\color{green}T_{ij}}$, ${\color{blue}T_{ij}T_{kj}}$, and ${\color{red}2T_{ij}T_{kj}}$, respectively. We refer the reader to Definition~\ref{def:canonical} for details of the labelling. 
	Here, $\bbT^*_3$ has 14 vertices and 33 edges.}
	\label{fig:minimal}
\end{figure}
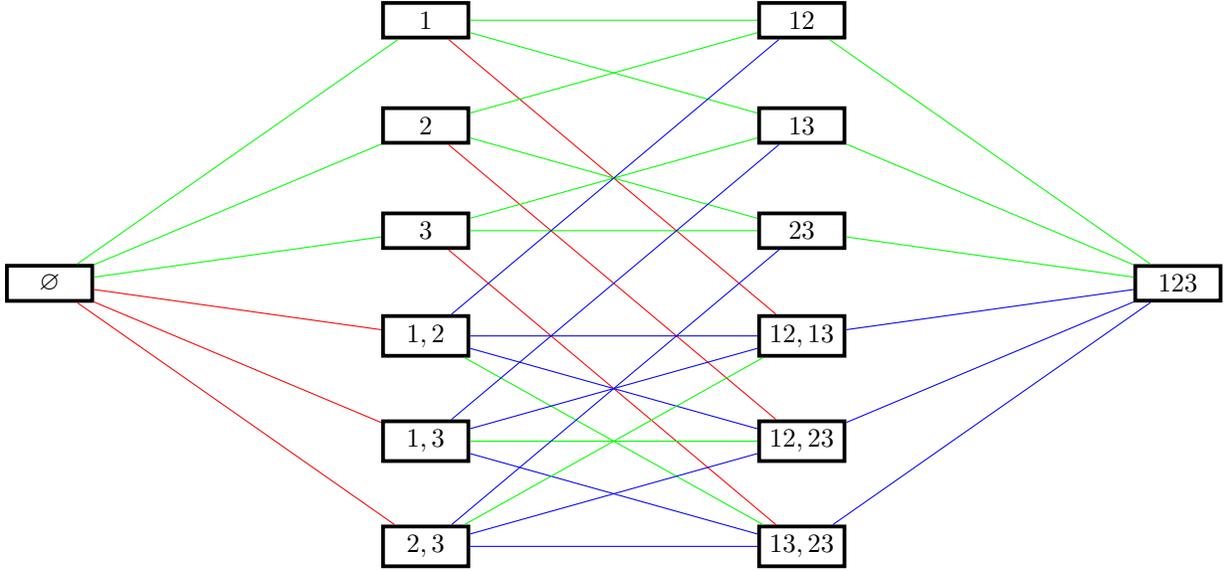

\begin{definition}[Trellis for Second-Order Permanents]\label{def:canonical}
Fix $M$ and let $\vT$ be an $M\times M$-matrix. 
Then the {\em second-order-permanent trellis $\bbT_M^*$} is defined as follows.
\begin{itemize}
\item (Vertices) For $0\le j\le M$, we set 
\begin{align*}
	\bbV^{(1)}_j &\triangleq \{ X \subset [M]: |X|=j \},\\
	\bbV^{(2)}_j & \triangleq \{ \{X_1,X_2\} : X_1, X_2 \subseteq [M], |X_1|=|X_2| = j \}\,.
\end{align*} and let $\bbV^{(1)}_j =\bbV_j \cup \bbV^{(2)}_j$. 
In other words, we have two types of vertices: the first type comprises all $\binom{M}{j}$ $j$-subsets of $[M]$; while the second type comprises all ${\binom{M}{j}\choose 2}$ unordered pairs of $j$-subsets.

\item (Edges and edge labels) First, we consider $X\in \bbV^{(1)}_{j-1}$ and so, $|X|=j-1$.
\begin{itemize}
	\item For $i\notin X$, we draw the edge from $X$ to $X\cup\{i\}$ and we label this edge with $T_{ij}$.
	\item For $i,k\notin X$ and $i< k$, we draw an edges $X$ $\{X\cup\{i\},X\cup\{k\}\}$, and  we label this edge with $2T_{ij}T_{kj}$.
\end{itemize}

Next, we consider $\{X_1,X_2\}\in \bbV^{(2)}_{j-1}$ and so, $|X_1|=|X_2|=j-1$.
\begin{itemize}
	\item Suppose that there exists $|X_3|=j$ such that $X_3=X_1\cup X_2$. Then we draw an edge from $\{X_1,X_2\}$ to $X_3$  and label this edge with $T_{ij}T_{kj}$, where 
	$\{i\} = X_3\setminus X_1$ and $\{k\} = X_3\setminus X_2$.
	\item If $i\ne k$, $i\notin X_1$ and $k\notin X_2$, we draw an edge from $\{X_1,X_2\}$ to $\{X_1\cup\{i\},X_2\cup\{k\}\}$ and  we label this edge with $T_{ij}T_{kj}$.
	\item If $i\notin X_1\cup X_2$, we draw an edge from $\{X_1,X_2\}$ to $\{X_1\cup\{i\},X_2\cup\{i\}\}$ and  we label this edge with $T_{ij}$.
\end{itemize}
\end{itemize}

\end{definition}

\begin{example}[Example~\ref{exa:nonminimal} continued]\label{exa:minimal}
	As before, set $M=3$ and we consider a $3\times 3$-matrix $\vT$. 
	Then second-order-permanent trellis $\bbT_3^*$ in Figure~\ref{fig:minimal}.
	We check that $\bbT_3^*$ has 14 vertices and 33 vertices.
	Hence, applying to the Viterbi algorithm on $\bbT_3^*$, we compute $\permsec(\vT)$ using only 27 multiplications and 20 additions. The number of arithmetic operations is significantly lesser than the number in Example~\ref{exa:nonminimal}.
\end{example}

\begin{proposition}
Let $\bbT_M^* = (\bbV^*,\bbE^*, \bbL^*)$ be as defined by Definition~\ref{def:canonical}.
Then we have that 
\begin{align}
	|\bbV^*| & = \frac12 \binom{2M}{M}+2^{M-1} \sim \frac{4^M}{2\sqrt{\pi M}}\,\\
	|\bbE^*| & \le \frac{M^2}{2} \binom{2M-2}{M-1} + M(M+1)2^{M-3} \sim \frac{M^{1.5}4^{M-1}}{2\sqrt{\pi}}\,.
\end{align}
\end{proposition}

\begin{proof}
First, we determine the total number of vertices. Now, for all $j$, we have that $\bbV^{(1)}_j=\binom{M}{j}$ and therefore, 
$\left|\bigcup_{j=0}^M \bbV^{(1)}_j \right| =\sum_{j=0}^M \binom{M}{j} = 2^M$.

On the other hand, for all $j$, we have that $\bbV^{(2)}_j=\binom{\binom{M}{j}}{2}$. So,

\begin{align*}
	\left|\bigcup_{j=0}^M \bbV^{(2)}_j \right| 
	&=\sum_{j=0}^M \frac{\binom{M}{j}(\binom{M}{j}-1)}{2}\\
	&= \frac12 \sum_{j=0}^M \binom{M}{j}^2-\binom{M}{j}\\
	&= \frac12 \sum_{j=0}^M \binom{M}{j}^2 - \frac12\sum_{j=0}^M\binom{M}{j}\\
	&= \frac12 \binom{2M}{M} - 2^{M-1}	
\end{align*}
Therefore, the total number of vertices is
$\frac12 \binom{2M}{M}+2^{M-1}$, as required.

\vspace{2mm}

Next, we determine the total number of edges. 
Now, for all $j$, a vertex in either $\bbV^{(1)}_j $ or $\bbV^{(2)}_j$ has out-degree at most $(M-j)^2$. 
Therefore,  the total number of edges exiting vertices in $\bigcup_{j=0}^M \bbV^{(1)}_j$ is at most 
\begin{equation*}
	\sum_{j=0}^{M} (M-j)^2 \binom{M}{j} = M(M+1) 2^{M-2}.
\end{equation*}

On the other hand, the total number of edges exiting vertices in $\bigcup_{j=0}^M \bbV^{(2)}_j$ is at most
\begin{align*}
	\sum_{j=0}^{M} (M-j)^2 \binom{\binom{M}{j}}{2} 
	& = \frac12 \sum_{j=0}^{M} (M-j)^2 \left(\binom{M}{j}^2-\binom{M}{j}\right)\\
	& = \frac{M^2}{2} \binom{2M-2}{M-1} - M(M+1)2^{M-3}.
\end{align*}

Therefore, we have at most
$\frac{M^2}{2} \binom{2M-2}{M-1} + M(M+1)2^{M-3}$ edges, as required. 
Throughout this proof, we have used a number of combinatorial identities obtained from \cite{oeis:A037966,oeis:A001788,oeis:A000984} (detailed proofs can be found in the references therein).
\end{proof}

Therefore, we have the following method that computes $\permsec(\vT)$ in $O(M^{1.5}4^{M-1})=o((M!)^2)$ time.

\begin{corollary}
Let $\vT$ be an $M\times M$-matrix.
We can compute $\permsec(\vT)$ in $O(M^{1.5}4^{M-1})$ time using the trellis $\bbT^*_M$.
\end{corollary}

\subsection{Simulation Results}

To end this section, we corroborate our estimates given by \eqref{upper}, \eqref{lower} and Theorem~\ref{thm:bound1} with numerical experiments. 
Specifically, we consider the binary erasure channel ${\rm BEC}(p)$ and estimate $\perror{\cC}$, 
where $\cC$ are the codes defined in Examples~\ref{exa:JEDI} and~\ref{exa:JLDI}.
Then using the methods described in Section~\ref{sec:computeV}, we determine the values of $U$ and $V$, and hence, obtain upper and lower bounds for $\perror{\cC}$ using \eqref{upper} and \eqref{lower}, respectively.
We also compute the upper bound provided by the closed formula in Theorem~\ref{thm:bound1}.
We then simulated 500,000 trials for various erasure probabilities $p$ and determined numerically the average failure rate.
The results are plotted in Figure~\ref{fig:simulations} and 
we observe that the estimates provided by \eqref{upper} and \eqref{lower} are very sharp.

\begin{figure}[!t]
\begin{center}
\begin{tabular}{p{80mm}p{80mm}}
(a)~Linear\,code of length\,five from in Example\,\ref{exa:JLDI}
& (b)~Simplex\,code of length\,seven from in Example\,\ref{exa:JEDI}
\\
\includegraphics[width=8cm]{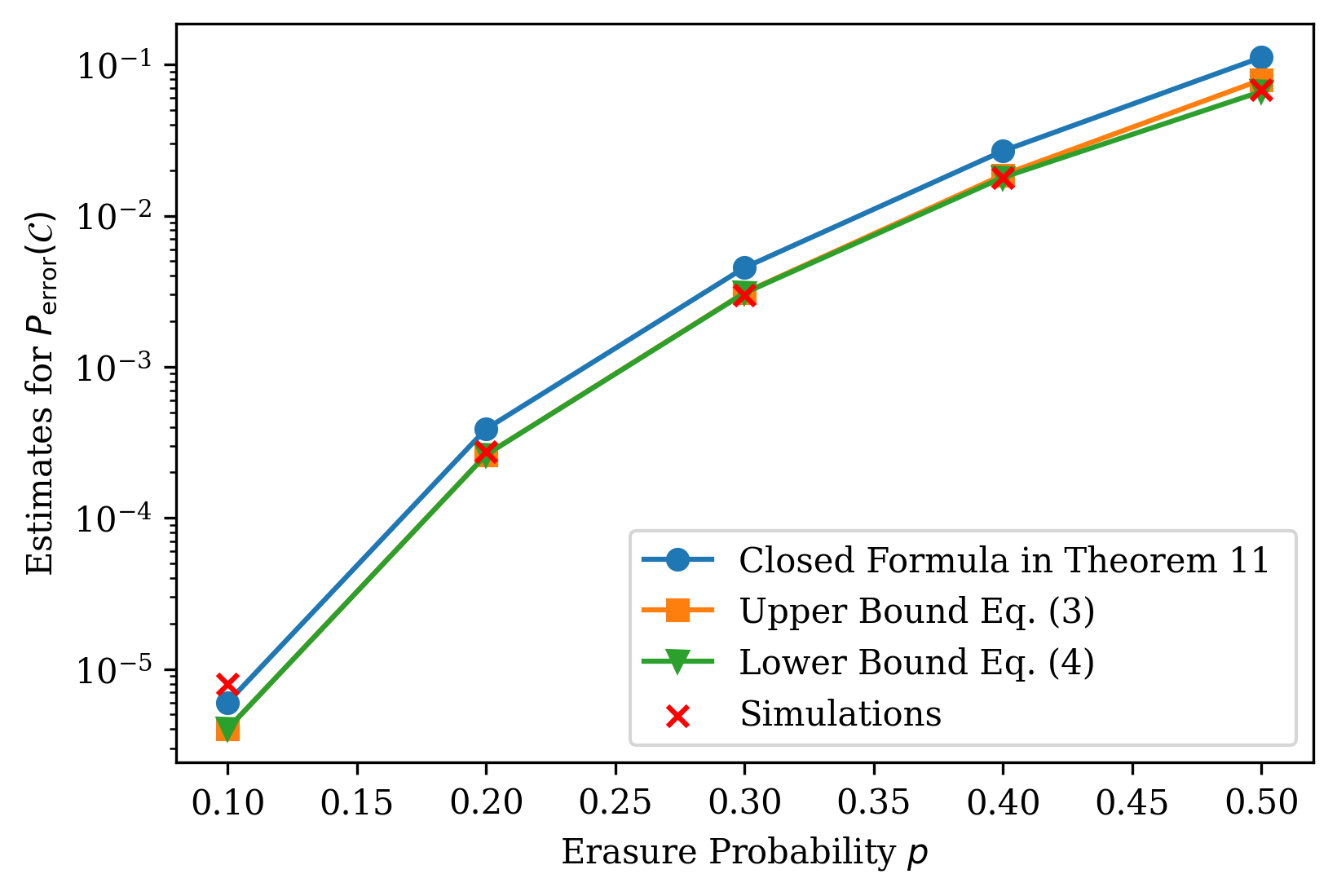} 
& \includegraphics[width=8cm]{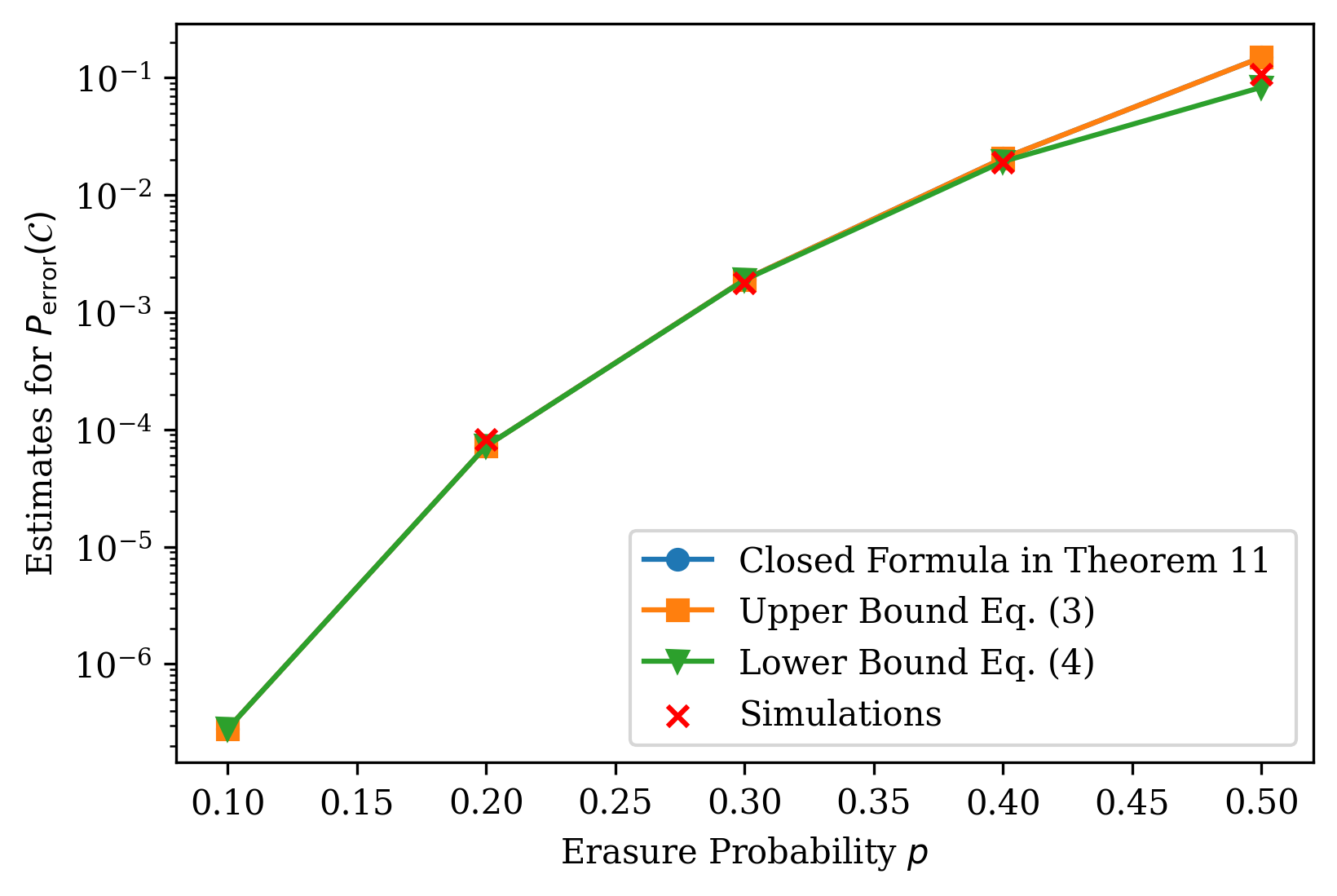}
	\end{tabular}
\end{center}
\caption{Estimates for the probability of erroneous identification $\perror{\cC}$ for various codes $\cC$}
\label{fig:simulations}
\end{figure}


\section{Discussion and Future Work}
\label{sec:discussion}

We discuss certain extensions and possible future work.

\begin{itemize}
	\item {\em General channels}. 
	As mentioned earlier, both JEDI and JMDI can extended to obtain a joint maximum likelihood decoder for other channels. 
	Specifically, suppose that a channel is described by a probability distribution $P$ where each output $\vy$ given an input $\vx$ is received with probability $P(\vy|\vx)$. 
	In Step 1 of JEDI~/~JMDI, we can create the bipartite input-output graph $\cG$ by drawing an edge $(\vx,\vy)$ whenever $P(\vy|\vx) > 0$, and then assigning the edge $(\vx,\vy)$ the cost $-\log P(\vy|\vx)$. 
	Then finding a minimum-cost perfect matching in the $\cG$ yields a permutation that maximizes the likelihood of correct identification.
	As with the analysis of JEDI, the time to find a minimum-cost perfect matching depends on the size of $\cG$, or the number of edges in $\cG$.
	In~\cite{Chrisnata2022}, Chrisnata \etal{} determined the expected number of edges in $\cG$ for both the insertion and deletion channels.
	It will be interest to study this quantity for other channels. 
	
	\item {\em Handling absentee bees}. In~\cite{TTV2020}, the authors studied the bee-identification problem for the scenario where bees were absent with certain probability. 
	In other words, instead of $M$ channel outputs, we have $M-a$ outputs where $a>0$.
	Both JEDI and JMDI can be modified to handle these scenarios. 
	In both cases, we proceed as before and simply add $a$ {\em absentee} right nodes to the bipartite graph $\cG$. 
	For the BEC case, we connect each absentee right node to all left nodes, 
	while for the BSC case, we connect each absentee right node to all left nodes and assign the cost to be zero.
	Then we find a perfect matching or a minimum-cost perfect matching as before.
	
	\item {\em Code Design}. In this paper, given a code $\cC$, we can construct the matrices $\TBEC$ and $\TBSCu$ as in Section~\ref{sec:prob}, and then estimate certain performance metrics of JEDI and JMDI. Specifically, the sum of entries of $\TBEC$ determines the running time of JEDI, while the permanents of $\TBEC$ and $\TBSCu$ provide upper bounds on the probability of erroneous identification.
	Alternatively, we can fix the probability of erroneous identification $\epsilon$, and {\em design} a code $\cC$ such that $\perror{\cC}\le \epsilon$. 
	This was partially studied in~\cite{TTV2019}. Specifically, for fixed $\epsilon$, Tandon \etal{} showed that there exists a random code $\cC$ with rate close to $R(\epsilon)$ and $\perror{\cC}\le \epsilon$. A natural question is to find an {\em explicit} family of codes that achieve the same property. 
	

\end{itemize}

\end{document}